\documentclass{article}

\usepackage[linkcolor=blue,citecolor=blue,colorlinks=true]{hyperref}
\usepackage[utf8]{inputenc}
\usepackage[a4paper,left=25mm,right=25mm]{geometry}
\usepackage{booktabs}
\usepackage{multirow}
\usepackage{amsmath,amssymb,amsthm}
\usepackage[ruled]{algorithm2e}
\usepackage{color}
\usepackage{graphicx}
\usepackage{subcaption}

\newtheorem*{definition*}{Definition}
\newtheorem{definition}{Definition}

\newtheorem{theorem}{Theorem}
\newtheorem{lemma}{Lemma}
\newtheorem*{lemma*}{Lemma}

\newtheorem{fact}{Fact}

\newcommand{\dist}{\operatorname{dist}}

\DeclareMathOperator*{\argmin}{arg\,min}
\DeclareMathOperator{\diversity}{div}
\DeclareMathOperator{\gendiv}{gen-div}
\DeclareMathOperator{\MST}{MST}
\DeclareMathOperator{\TSP}{TSP}

\renewcommand{\epsilon}{\varepsilon}
\newcommand{\BO}[1]{O\left( #1 \right)}
\newcommand{\BT}[1]{\Theta\left( #1 \right)}
\newcommand{\BOM}[1]{\Omega\left( #1 \right)}
\newcommand{\Let}[2]{#1 $\leftarrow$ #2} \SetKw{To}{to}

\title{
  MapReduce and Streaming Algorithms for Diversity Maximization in 
  Metric Spaces of Bounded Doubling Dimension\footnote{This work was published in the Proocedings of the VLDB Endowment~\cite{CeccarelloPPU17}.}
}

\date{}

\newcommand{\affaddr}[1]{{\small #1}}  % author affiliation line
\newcommand{\email}[1]{{\tt\small #1}} % author email

\author{
  Matteo Ceccarello$^1$\and
  Andrea Pietracaprina$^1$ \and
  Geppino Pucci$^1$\and
  Eli Upfal$^2$
  \\
  $^1$\affaddr{Department of Information Engineering, University of Padova,
    Padova, Italy}\\
  \email{\{ceccarel,capri,geppo\}@dei.unipd.it}
  \\
  $^2$\affaddr{Department of Computer Science,
    Brown University, Providence, RI USA}\\
  \email{eli\_upfal@brown.edu}
}

\begin{document}

\maketitle

\begin{abstract}
Given a dataset of points in a metric space and an integer $k$, a
diversity maximization problem requires determining a subset of $k$
points maximizing some diversity objective measure, e.g., the minimum
or the average distance between two points in the subset.
Diversity maximization is computationally hard, hence only approximate
solutions can be hoped for. Although its applications are mainly in
massive data analysis, most of the past research on diversity
maximization focused on the sequential setting.  In
this work we present space and pass/round-efficient 
diversity maximization algorithms for the
Streaming and MapReduce models and analyze their approximation
guarantees for the relevant class of metric spaces of bounded doubling
dimension. Like other approaches in the literature, our algorithms
rely on the determination of high-quality core-sets, i.e.,
(much) smaller subsets of the input which contain good approximations
to the optimal solution for the whole input. For a variety of
diversity objective functions, our algorithms attain an
$(\alpha+\epsilon)$-approximation ratio, for any constant
$\epsilon>0$, where $\alpha$ is the best approximation ratio achieved
by a polynomial-time, linear-space sequential algorithm for the same
diversity objective.  This improves substantially  over
the approximation ratios attainable in Streaming and MapReduce
by state-of-the-art algorithms for general metric
spaces. We provide extensive experimental evidence of the
effectiveness of our algorithms on both
real world and synthetic datasets, scaling up to over a
billion points.
\end{abstract}

\section{Introduction}

\emph{Diversity maximization} is a fundamental primitive in massive
data analysis, which provides a succinct summary 
%(whose size is regulated by a parameter $k$) 
of a dataset while preserving the diversity of the
data~\cite{AbbassiMT13,MasinB08,wu2013,YangMNFCH15}. This summary can be presented visually to the user or can be used as a core for further processing of the dataset. In this paper we present novel efficient algorithms for diversity
maximization in popular computation models for massive data processing,
namely Streaming and MapReduce.

\paragraph{Diversity Measures and their Applications:}
Given a dataset of
points in a metric space and a constant $k$, a solution to the
diversity maximization problem is a subset of $k$ points that
maximizes some diversity objective measure defined in terms of the
distances between the points.

Combinations of relevance ranking and diversity maximization have been
explored in a variety of applications, including web
search~\cite{AngelK11}, e-commerce~\cite{BhattacharyaGM11},
recommendation systems~\cite{YuLA09}, aggregate
websites~\cite{MunsonZR09} and query-result navigation~\cite{ChenL07} (see
\cite{RosenkrantzRT07,AbbassiMT13,IndykMMM14} for further references on the
applications of diversity maximization).  The common problem in all
these applications is that even after filtering and ranking for
relevance, the output set is often too large to be presented to the
user. A practical solution is to present a diverse subset of the
results so the user can evaluate the variety of options and possibly
refine the search.

There are a number of ways to formulate the goal of finding a set of
$k$ points which are as diverse, or as far as possible from each
other. Conceptually, a $k$-diversity maximization problem can be
formulated in terms of a specific graph-theoretic measure defined on
sets of $k$ points, seen as the nodes of a clique where each edge is
weighted with the distance between its endpoints~\cite{ChandraH01}.
Several diversity measures are defined in
Table~\ref{tab:diversity-notions}.  While the most appropriate ones in
the context of web search, e-commerce, aggregator systems and query results
navigation are the remote-edge and the remote-clique measures
\cite{GollapudiS09,AbbassiMT13}, the results in this paper also extend
to the other measures in the table, which have important applications
in analyzing network performance, locating strategic facilities or
noncompeting franchises, or determining initial solutions for
iterative clustering algorithms or heuristics for hard optimization
problems such as
TSP~\cite{HalldorssonIKT99,ChandraH01,RosenkrantzRT07}. We include all
of these measures here to demonstrate the versatility of our approach
to a variety of diversity criteria. We want to stress that different
measures characterize the diversity of a set in a different fashion:
indeed, an optimal solution with respect to one measure is not
necessarily optimal with respect to another measure.

\paragraph{Distance Metric:} All the diversity criteria listed in
Table~\ref{tab:diversity-notions} are known to be NP-hard for general
metric spaces. Following a number of recent
works~\cite{AckermannBS10,ColeL06,KonjevodRX08,GottliebKK14,CeccarelloPPU15,CeccarelloPPU16},
we parameterize our results in terms of the \emph{doubling dimension}
of the metric space. Recall that a metric space has doubling dimension
$D$ if any ball of radius $r$ can be covered by at most $2^D$ balls of
radius $r/2$. While our methods yield provably tight bounds in spaces
of bounded doubling dimension (e.g., any bounded dimension Euclidian space) they have the ability of providing good approximations
in more general spaces based on important practical distance functions
such as the cosine distance in web search~\cite{AngelK11} and the dissimilarity
(Jaccard) distance in database queries~\cite{LeskovecRU14}.

\paragraph{Massive Data Computation Models:}
Since the applications of diversity maximization are
mostly in the realm of massive data analysis, it is important to
develop efficient algorithms for computational settings that can
handle very large datasets.  
\newcommand{\definitionTableVerticalSpacing}{\rule{0pt}{10pt}}
\begin{table}[t]
  \centering
  \begin{tabular}{llc}%{l@{\hskip 2pt} l@{\hskip 1pt} c}
    \toprule
    Problem
     & Diversity measure
     & Sequential approximation
    \\
    \midrule
    remote-edge
      & $\min_{p, q\in S} d(p, q)$
      & 2 (2)~\cite{Tamir91} 
    \\
    \definitionTableVerticalSpacing%
    remote-clique
      & $\sum_{p, q\in S} d(p, q)$ 
      & 2 ($-$)~\cite{HassinRT97} 
    \\
    \definitionTableVerticalSpacing%
    remote-star 
      & $\min_{c\in S}\sum_{q\in S\setminus\{c\}} d(c, q)$ 
      & 2 ($-$)~\cite{ChandraH01} 
    \\
    \definitionTableVerticalSpacing%
    remote-bipartition
      & $\min_{Q\subset S, |Q|=\lfloor|S|/2\rfloor}\sum_{q\in Q, z\in S\setminus Q} d(q, z)$
     % & $\displaystyle\min_{{\substack{\scriptscriptstyle Q\subset S \\ \scriptscriptstyle |Q|=\lfloor|S|/2\rfloor}}}{\textstyle\sum}_{\substack{q\in Q\\ z\in S\setminus Q}} d(q, z)$ 
     & 3 ($-$)~\cite{ChandraH01} 
    \\
    \definitionTableVerticalSpacing%
    remote-tree 
      & $w(\MST(S))$ 
      & 4 (2)~\cite{HalldorssonIKT99} 
    \\
    \definitionTableVerticalSpacing%
    remote-cycle 
      & $w(\TSP(S))$
      & 3 (2)~\cite{HalldorssonIKT99} 
    \\
    \bottomrule
  \end{tabular}
  \caption{% 
    Diversity measures considered in this paper.  $w(\MST(S))$ (resp.,
    $w(\TSP(S))$) denotes the minimum weight of a spanning tree (resp.,
    Hamiltonian cycle) of the complete graph whose nodes are the points of
    $S$ and whose edge weights are the pairwise distances among the
    points.  The last column lists the best known approximation factor,
    the lower bound under the hypothesis $P\neq NP$ (in parentheses), and
    the related references. 
  }\label{tab:diversity-notions}
\end{table}
The Streaming and MapReduce models are widely recognized as suitable
computational frameworks for big-data processing.  The Streaming
model~\cite{RaghavanH99} copes with large data volumes through an
on-the-fly computation on the streamlined dataset, storing only very
limited information in the process, while the MapReduce
model~\cite{KarloffSV10,PietracaprinaPRSU12} enables the handling of
large datasets through the massive availability of resource-limited
processing elements working in parallel. The major challenge in
both models is devising strategies which work under the constraint
that the number of data items that a single processor can access
simultaneously is substantially limited.

\iffalse More specifically, in the
Streaming model we have a single processor with memory space that is
significantly smaller than the total size of the data stream, while,
in MapReduce, the local memory of each reducer (processor) is
polynomially sublinear in the total data size. An obvious question is
how this limitations affect the quality of the approximation for the
various diversity problems. In this work, we provide an answer to this
question, focusing on the case of metric spaces of bounded doubling
dimension, a class of metric spaces that has attracted significant
interest in the context of data analysis, and that includes the
important family of Euclidean spaces of constant dimension.  
\fi

%In this paper we present novel efficient algorithms for diversity
%maximization in the Streaming and MapReduce models, focusing on the
%relevant case of metric spaces of bounded doubling dimension, that is, metric
%spaces where any ball of radius $r$ can be covered by a constant
%number of balls of radius $r/2$. Among others, this class of metric spaces
%includes the important family of Euclidean spaces of constant
%dimension.

\paragraph{Related work.}
\sloppy
Diversity maximization has been studied in the literature under
different names (e.g., $p$-Dispersion, Max-Min Facility Dispersion,
etc.).  An extensive account of the existing formulations is provided
in~\cite{ChandraH01}. All of these problems are known to be NP-hard,
and several sequential approximation algorithms have been proposed.
Table~\ref{tab:diversity-notions} summarizes the best known results
for general metric spaces.  There are also some specialized results
for spaces with bounded doubling dimension: for the remote-clique
problem, a polynomial-time $(\sqrt{2}+\epsilon)$-approximation
algorithm on the Euclidean plane, and a polynomial-time
$(1+\epsilon)$-approximation algorithm on $d$-dimensional spaces with
rectilinear distances, for any positive constants $\epsilon > 0$ and
$d$, are presented in~\cite{FeketeM04}.  In~\cite{HalldorssonIKT99} it
is shown that a natural greedy algorithm attains a $2.309$
approximation factor on the Euclidean plane for remote-tree. Recently,
the remote-clique problem has been considered under matroid
constraints~\cite{AbbassiMT13,CevallosEZ16}, which generalize the
cardinality constraints considered in previous literature.

In recent years, the notion of (composable) core-set has been
introduced as a key tool for the efficient solution of optimization
problems on large datasets.  A \emph{core-set}~\cite{AgarwalHV05},
with respect to a given computational objective, is a (small) subset
of the entire dataset which contains a good approximation to the
optimal solution for the entire dataset. A \emph{composable
  core-set}~\cite{IndykMMM14} is a collection of core-sets, one for
each subset in an arbitrary partition of the dataset, such that the
union of these core-sets contains a good core-set for the entire
dataset. The approximation factor attained by a (composable) core-set
is defined as the ratio between the value of the global optimal
solution and the value of the optimal solution on the (composable)
core-set. For the problems listed in
Table~\ref{tab:diversity-notions}, composable core-sets with constant
approximation factors have been devised in
\cite{IndykMMM14,AghamolaeiFZ15} (see
Table~\ref{tab:core-set-approximations}).  As observed
in~\cite{IndykMMM14}, (composable) core-sets may become key
ingredients for developing efficient algorithms for the MapReduce and
Streaming frameworks, where the memory available for a processor's
local operations is typically much smaller than the overall input
size.

In recent years, the characterization of data through the doubling
dimension of the space it belongs to has been increasingly used for
algorithm design and analysis in a number of contexts, including
clustering~\cite{AckermannBS10}, nearest neighbour
search~\cite{ColeL06}, routing~\cite{KonjevodRX08}, machine
learning~\cite{GottliebKK14}, and graph
analytics~\cite{CeccarelloPPU15, CeccarelloPPU16}.
%
%In~\cite{GottliebK13}, the authors define a \emph{nearly doubling} set
%of points as a set $S$ where all but $\sqrt{|S|}$ points have constant
%doubling dimension. They provide an approximation algorithm for the
%problem of removing from a set of points the fewest number of points
%in order to get to a target doubling dimension $d^*$. Moreover, they
%give a 2-approximation algorithm to compute the doubling dimension of
%a point set (which is NP-hard).
%

\begin{table}
  \centering
  \begin{tabular}{lcc}%{l@{\hskip 0pt}cc}
    \toprule
     & Previous~\cite{IndykMMM14,AghamolaeiFZ15}
     & Our results
    \\
     & {\footnotesize General metric spaces}
     & {\footnotesize Bounded doubling dimension}
    \\
    \midrule
    remote-edge
     & $3$
     & $1+\epsilon$
    \\
    remote-clique
    % The tradeoff of the epsilon here is with the time.
     & $6+\epsilon$ 
     & $1+\epsilon$
    \\
    remote-star
     & $12$ 
     & $1+\epsilon$
    \\
    remote-bipartition
     & $18$
     & $1+\epsilon$
    \\
    remote-tree
     & $4$ 
     & $1+\epsilon$
    \\
    remote-cycle
     & $3$ 
     & $1+\epsilon$
    \\
    \bottomrule
  \end{tabular}
  \caption{%
    Approximation factors of the composable core-sets computed by our algorithm,
    compared with previous approaches.%
  }\label{tab:core-set-approximations}
\end{table}

\paragraph{Our contribution.}
In this paper we develop efficient algorithms for diversity
maximization in the Streaming and MapReduce models.  At the heart of
our algorithms are novel constructions of (composable) core-sets.
In contrast to~\cite{IndykMMM14, AghamolaeiFZ15}, where
different constructions are devised for each diversity objective, we
provide a unique construction technique for all of the six objective
functions.  While our approach is applicable to general metric spaces,
on spaces of bounded doubling dimension, our (composable) core-sets
feature a $1+\epsilon$ approximation factor, for any fixed
$0<\epsilon\leq 1$, for all of the six diversity objectives, with the
core-set size increasing as a function of
$1/\epsilon$. The approximation factor is significantly better than
the ones attained by the known composable core-sets in general metric
spaces, which are reported in Table~\ref{tab:core-set-approximations}
for comparison.

Once a core-set (possibly obtained as the union of composable
core-sets) is extracted from the data, the best known sequential
approximation algorithm can be run on it to derive the final
solution. The resulting approximation ratio attained in this fashion
combines two sources of error: (1) the approximation loss in replacing
the entire dataset with a core-set; and (2) the approximation factor
of the sequential approximation algorithm executed on the core-set.
On metric spaces of bounded doubling dimension the combined
approximation ratio attained by our algorithms for any of the six
diversity objective functions considered in the paper is bounded by
$(\alpha+\epsilon)$, for any constant $0<\epsilon\leq 1$, where
$\alpha$ the is best approximation ratio achieved by a
polynomial-time, linear-space sequential algorithm for the same
maximum diversity criterion.

Our algorithms require only one pass over the data, in the streaming
setting, and only two rounds in MapReduce. To the best of our
knowledge, for all six diversity problems, our streaming algorithms
are the first ones that yield approximation ratios close to those of
the best sequential algorithms using space independent of input stream
size. Also, we remark that the parallel strategy at the base of the
MapReduce algorithms can be effectively ported to other models of
parallel computation.

Finally, we provide experimental evidence of the practical relevance
of our algorithms on both synthetic and real-world datasets. In
particular, we show that higher accuracy is achievable by increasing
the size of the core-sets, and that the MapReduce algorithm is
considerably faster (up to three orders of magnitude) than its
state-of-the-art competitors.  Also, we provide evidence that the
proposed approach is highly scalable. We want to remark that our work
provides the first substantial experimental study on the performance
of diversity maximization algorithms on large instances of up to
billions of data points.

The rest of the paper is organized as follows.  In
Section~\ref{sec:preliminaries}, we introduce some fundamental
concepts and useful notations. In Section~\ref{sec:basic-properties},
we identify sufficient conditions for a subset of points to be a
core-set with provable approximation guarantees. These properties are
then crucially exploited by the streaming and MapReduce algorithms
described in Sections~\ref{sec:streaming} and~\ref{sec:mapreduce},
respectively. Section~\ref{sec:generalized} discusses how the higher
memory requirements of four of the six diversity problems can be
reduced, while Section~\ref{sec:experiments} reports on the results of
the experiments.

\section{Preliminaries} \label{sec:preliminaries}

Let $({\cal D}, d)$ be a metric space. The distance between two points
$u, v \in {\cal D}$ is denoted by $d(u, v)$.  Moreover, we let
$d(p, S)=\min_{q\in S} d(p, q)$ denote the minimum distance between a
point $p\in {\cal D}$ and an element of a set $S\subseteq {\cal D}$. Also,
for a point $p\in{\cal D}$, the \emph{ball of radius $r$ centered at
  $p$} is the set of all points in ${\cal D}$ at distance at most $r$
from $p$.
The \emph{doubling dimension} of a space is the smallest $D$ such
that any ball of radius $r$ is covered by  at most $2^D$ balls of radius
$r/2$~\cite{GuptaKL03}. As an immediate consequence, for any 
$0<\epsilon\le 1$, any ball of radius $r$ can be covered by at most $(1/\epsilon)^D$ balls
of radius $\epsilon r$. 
For ease of presentation, in this paper we concentrate on metric
spaces of constant doubling dimension $D$, although the results can be
immediately extended to nonconstant $D$ by suitably adjusting the
ranges of variability of the parameters involved. Several relevant
metric spaces have constant doubling dimension, a notable case being 
Euclidean space of  constant dimension $D$, which has doubling dimension
$\BO{D}$~\cite{GuptaKL03}.

Let $\diversity: 2^{\cal D} \rightarrow \mathbb{R}$ be a \emph{diversity
  function} that maps a set $S\subset {\cal D}$ to some nonegative real
number.
In this paper, we will consider the instantiations of function
$\diversity$ listed in Table~\ref{tab:diversity-notions}, which were
introduced and studied in~\cite{ChandraH01,IndykMMM14,AghamolaeiFZ15}.
For a specific diversity function $\diversity$, a set $S\subset
{\cal D}$ of size $n$ and a positive integer $k \leq n$, 
the goal of the \emph{diversity
  maximization problem} is to find some subset $S'\subseteq S$ of size
$k$ that maximizes the value $\diversity(S')$. In the following, we
refer to the \emph{$k$-diversity} of $S$ as
\[
  \diversity_k(S) = \max_{S'\subseteq S, |S'| = k} \diversity( S')
\]

The notion of \emph{core-set}~\cite{AgarwalHV05} captures the idea of
a small set of points that approximates some property of a larger
set.
\begin{definition}\label{def:composable}
Let $\diversity(\cdot)$ be a diversity function, $k$ be a positive
  integer, and $\beta \ge 1$. A set $T\subseteq S$, with $|T| \ge
  k$, is a {\em $\beta$-core-set} for $S$ if
  \[
    \diversity_k(T) \ge \frac{1}{\beta} \diversity_k(S)
  \]
\end{definition}

In~\cite{IndykMMM14, AghamolaeiFZ15}, the concept of core-set is
extended so that, given an arbitrary partition of the input set, the
union of the core-sets of each subset in the partition is a core-set
for the entire input set.
\begin{definition}\label{def:composable-core-set}
  Let $\diversity(\cdot)$ be a diversity function, $k$ be a positive
  integer, and $\beta \ge 1$. A function $c(S)$ that maps
  $S\subset {\cal D}$ to one of its subsets computes a {\em
    $\beta$-composable core-set} w.r.t.  $\diversity$ if, for any
  collection of disjoint sets $S_1,\dots, S_\ell \subset {\cal D}$ with
  $|S_i| \ge k$, we have
  \[
    \diversity_k\left(\bigcup_{i=1}^\ell c(S_i)\right) \ge
    \frac{1}{\beta} \diversity_k\left(\bigcup_{i=1}^\ell S_i\right)
  \]
\end{definition}

Consider a set $S \subseteq {\cal D}$ and a subset $T \subseteq S$. We
define the \emph{range} of $T$ as
$r_T =\max_{p\in S\setminus T} d(p, T)$, and the \emph{farness} of $T$
as $\rho_T= \min_{c\in T}\{d(c, T\setminus \{c\})\}$. Moreover, we
define the \emph{optimal range} $r_k^*$ for $S$ w.r.t. $k$ to be the
minimum range of a subset of $k$ points of $S$. Similarly, we define
the \emph{optimal farness} $\rho_k^*$ for $S$ w.r.t. $k$ to be the
maximum farness of a subset of $k$ points of $S$. Observe that
$\rho_k^*$ is also the value of the optimal solution to the
remote-edge problem.

\section{Core-set characterization}\label{sec:basic-properties}

In this section we identify some properties that, when exhibited by a
set of points, guarantee that the set is a $(1+\epsilon)$-core-set for
the diversity problems listed in Table~\ref{tab:diversity-notions}. In
the subsequent sections we will show how core-sets with these
properties can be obtained in the streaming and MapReduce settings. In
fact, when we discuss the MapReduce setting, we will also show that
these properties also yield composable core-sets featuring tighter
approximation factors than existing ones, for spaces
with bounded doubling dimension.

First, we need to establish a fundamental relation between the optimal
range $r_k^*$ and the optimal farness $\rho_k^*$ for a set $S$.  To
this purpose, we observe that the classical greedy approximation
algorithm proposed in \cite{Gonzalez85} for finding a subset of
minimum range (\emph{$k$-center problem}), gives in fact a good
approximation to both measures. We refer to this algorithm as {\sc
  GMM}.  Consider a set of points $S$ and a positive integer $k <
|S|$. Let $T=\mbox{\sc GMM}(S,k)$ be the subset of $k$ points returned
by the algorithm for this instance. The algorithm initializes $T$ with
an arbitrary point $a\in S$. Then, greedily, it adds to $T$ the point
of $S\setminus T$ which maximizes the distance from the already
selected points, until $T$ has size $k$.  It is known that the
returned set $T$ is such that $r_T \leq 2 r_k^*$~\cite{Gonzalez85} and
it is easily seen that $r_T \leq \rho_T$ (referred to as
\emph{anticover property}). This immediately implies the following
fundamental relation.
\begin{fact}\label{fact:opt-radius-farness}
Given a set $S$ and $k>0$, we have $r_k^* \le
  \rho_k^*$.
\end{fact}

Let $S$ be a set belonging to a metric space of doubling dimension
$D$. In what follows, $\diversity(\cdot)$ denotes the diversity
function of the problem under consideration, and $O$ denotes an
optimal solution to the problem with respect to instance $S$.
Consider a subset $T\subseteq S$. Intuitively, $T$ is a good core-set
for some diversity measure on $S$, if for each point of the optimal
solution $O$ it contains a point sufficiently close to it.  We
formalize this intuition by suitably adapting the notion of
\emph{proxy function} introduced in \cite{IndykMMM14}. Given a
core-set $T\subseteq S$, we aim at defining a function
$p: O \rightarrow T$ such that the distance between $o$ and $p(o)$ is
bounded, for any $o\in O$. For some problems this function will be
required to be injective, whereas for, some others, injectivity will not
be needed. We begin by studying the remote-edge and the remote-cycle problem.

\begin{lemma}\label{lem:remote-edge}
  For any given $\epsilon > 0$, let $\epsilon'$ be such that
  $(1-\epsilon')=1/(1+\epsilon)$. A set $T\subseteq S$ is a
  $(1+\epsilon)$-core-set for the remote-edge and the remote-cycle problems if
  $|T| \ge k$ and there is a function $p: O\rightarrow T$ such that, for any
  $o\in O$, $d(o, p(o)) \le (\epsilon'/2)\rho_k^*$.
\end{lemma}
\begin{proof}
  Consider the remote-edge problem first, and observe that 
  $\diversity_k(T)\leq\diversity(O)=\rho_k^*$. By
  applying the triangle inequality and the stated property of the proxy 
  function $p$ we get
  \[
    \begin{aligned}
      \diversity_k(T)
      &\ge \min_{o_1, o_2 \in O} d(p(o_1), p(o_2))
      \ge \min_{o_1, o_2 \in O} \{
      d(o_1, o_2) - d(o_1, p(o_1)) - d(o_2, p(o_2))
      \} \\
      &\ge \min_{o_1, o_2 \in O} d(o_1, o_2) -\epsilon'\rho_k^* 
      = \diversity(O)(1-\epsilon') = \diversity(O)/(1+\epsilon)
    \end{aligned}
  \]
  Note that $p(\cdot)$ does not need to be injective: in fact, if two
  points of the optimal solution are mapped into the same proxy, the
  first inequality trivially holds, its right hand side being zero.

  Consider now the remote-cycle problem.
  Note that $\diversity_k(T) \leq \diversity(O)$. Let $\bar{\rho} =
  \diversity(O)/k$ and observe that $\rho_k^* \le \bar{\rho}$.
  Let $P=\{p(o) : o\in O\} \subseteq T$ be the image of the proxy
  function. Following the argument given in
  \cite{IndykMMM14,AghamolaeiFZ15}, consider $\operatorname{TSP}(P)$,
  an optimal tour on $P$. We build a weighted graph $G$ whose vertex
  set is $O\cup P$ and whose edges are those induced by
  $\operatorname{TSP}(P)$ plus two copies of edge $(o, p(o))$, for
  each $o\in O$. The weight of an edge $(u,v)$ is $d(u,v)$.  Clearly,
  the resulting graph $G$ is connected and all its vertices have even
  degree, therefore it admits an Euler tour $T_E$ of its edges. From
  $T_E$ we obtain a cycle $C$ of $O$ by shortcutting all nodes that
  are not in $O$. By repeated applications of the triangle inequality
  during shortcutting and the fact that
  $d(o, p(o)) \le (\epsilon'/2) \bar\rho$, we obtain:
  \[
    \begin{aligned}
      w(\operatorname{TSP}(O))
      &\le w(C) \le w(T_E)
      \le w(\operatorname{TSP}(P)) + k\epsilon' \bar{\rho}\\
      &\le \diversity_k(T) + \epsilon' \diversity(O)
    \end{aligned}
  \]
  Therefore,
  $\diversity(O) \le \diversity_k(T)/(1-\epsilon') =
  \diversity_k(T)(1+\epsilon)$.
  As in the case of the remote-edge problem, the injectivity of $p(\cdot)$ is not necessary.
\end{proof}

Note that the proof of the above lemma does not require $p(\cdot)$ to
be injective. Instead, injectivity is required for the remote-clique,
remote-star, remote-bipartition, and remote-tree problems, which are 
considered next.
\begin{lemma}\label{lem:remote-csbt}
For a given $\epsilon > 0$, let $\epsilon'$ be such that
$1-\epsilon'=1/(1+\epsilon)$. A set $T\subseteq S$ is a
$(1+\epsilon)$-core-set for the remote-clique, remote-star,
remote-bipartition, and remote-tree problems if
$|T| \ge k$ and  there is an injective function $p: O \rightarrow T$ such that,
for any $o\in O$, $d(o, p(o))\le (\epsilon'/2) \rho_k^*$.
\end{lemma}
\begin{proof}
  Observe that for each of the four problems it holds that
  $\diversity_k(T) \leq \diversity(O)$. Let us consider the
  remote-clique problem first, and define
  $\bar{\rho} = \diversity(O)/{k\choose 2} =
  \sum_{o_1, o_2 \in O}d(o_1, o_2)/{k\choose 2}$
  Clearly, $\rho_k^*\le \bar{\rho}$. By combining this observation with the
  triangle inequality we have
  \[
    \begin{aligned}
      \diversity_k(T') 
      &\ge \sum_{o_1, o_2\in O} d(p(o_1), p(o_2))
      \ge \sum_{o_1, o_2\in O} \left[ 
        d(o_1, o_2) - d(o_1, p(o_1)) - d(o_2, p(o_2))
      \right] \\
      &\ge \sum_{o_1, o_2} d(o_1, o_2) - {k\choose
        2}\epsilon'\bar{\rho} = \diversity(O)/(1+\epsilon) 
    \end{aligned}
  \]
  The injectivity of $p(\cdot)$ is needed in this case for the first
  inequality above to be true, since $k$ distinct proxies are needed to get a
  feasible solution. The argument for the other problems is virtually
  identical, and we omit it for brevity.
\end{proof}

\section{Applications to data streams}
\label{sec:streaming}

In the Streaming model \cite{RaghavanH99} one processor with a
limited-size main memory is available for the computation. The input
is provided as a continuous stream of items which is typically too
large to fit in main memory, hence it must be processed on the fly
within the limited memory budget.  Streaming algorithms aim at
performing as few passes as possible (ideally just one) over the
input.

In~\cite{IndykMMM14}, the authors propose the following use of
composable core-sets to approximate diversity in the streaming
model. The stream of $n$ input points is partitioned into $\sqrt{n/k}$
blocks of size $\sqrt{kn}$ each, and a core-set of size $k$ is
computed from each block and kept in memory. At the end of the pass,
the final solution is computed on the union of the core-sets, whose
total size is $\sqrt{kn}$. In this section, we show that substantial
savings (a space requirement independent of $n$) can be obtained by
computing a \emph{single} core-set from the entire stream through two
suitable variants of the 8-approximation \emph{doubling algorithm} for
the $k$-center problem presented in~\cite{CharikarCFM04}, which are
described below.

Let $k,k'$ be two positive integers, with $k \le k'$. The first
variant, dubbed {\sc SMM}$(S, k, k')$, works in phases and maintains in memory
a set $T$ of at most $k' + 1$ points. Each Phase~$i$ is associated
with a distance threshold $d_i$, and is divided into a \emph{merge
  step} and an \emph{update step}. Phase~1 starts after an
initialization in which the first $k'+1$ points of the stream are
added to $T$, and $d_1$ is set equal to
$\min_{c\in T}d(c, T \setminus\{c\})$. At the beginning of Phase~$i$,
with $i \geq 1$, the following invariant holds. Let $S_i$ be the
prefix of the stream processed so far. Then:
\begin{enumerate}
\item $\forall p \in S_i$, $d(p, T) \le 2d_i$
\item $\forall t_1, t_2 \in T$, with $t_1 \neq t_2$, we have $d(t_1,
  t_2) \ge d_i$
\end{enumerate}

Observe that the invariant holds at the beginning of Phase~1.  The
merge step operates on a graph $G=(T,E)$ where there is an edge
$(t_1,t_2)$ between two points $t_1 \neq t_2 \in T$ if
$d(t_1, t_2) \le 2d_i$. In this step, the algorithm seeks a maximal
independent set $I \subseteq T$ of $G$, and sets $T=I$. The update
step accepts new points from the stream. Let $p$ be one such new
point. If $d(p, T) \le 4d_i$, the algorithm discards $p$, otherwise it
adds $p$ to $T$. The update step terminates when either the stream
ends or the $(k'+1)$-st point is added to $T$. At the end of the step,
$d_{i+1}$ is set equal to $2d_i$. As shown in~\cite{CharikarCFM04}, at
the end of the update step, the set $T$ and the threshold $d_{i+1}$
satisfy the above invariants for Phase~$i+1$.

To be able to use {\sc SMM} for computing a core-set for our diversity
problems, we have to make sure that the set $T$ returned by the
algorithm contains at least $k$ points. However, in the algorithm
described above the last phase could end with $|T| < k$. To fix this
situation, we modify the algorithm so to retain in memory, for the
duration of each phase, the set $M$ of points that have been removed
from $T$ during the merge step performed at the beginning of the
phase. Consider the last phase. If at the end of the stream we have
$|T| < k$, we can pick $k - |T|$ arbitrary nodes from $M$ and add them
to $T$. Note that we can always do so because
$M \cup I = k' + 1 \ge k$, where $I$ is the independent set found
during the last merge step.

Suppose that the input set $S$ belongs to a metric space with doubling
dimension $D$. We have:
\begin{lemma}\label{lem:smm-properties}
  For any $0 < \epsilon' \leq 1$, let $k'=(32/\epsilon')^D\cdot k$,
  and let $T$ be the set of points returned by {\sc SMM}$(S, k,
  k')$. Then, given an arbitrary set $X\subseteq S$ with $|X|=k$,
  there exist a function $p: X\rightarrow T$ such that, for any
  $x\in X$, $d(x, p(x))\le (\epsilon'/2)\rho^*_k$.
\end{lemma}
\begin{proof}
  Let $r_{k'}^*$ to be the optimal range for $S$ w.r.t. $k'$.  Also,
  let $r_{T}=\max_{p\in S} d(p, T)$ be the range of $T$ and let
  $\rho_k^*$ be the optimal farness for $S$ w.r.t. $k$.  Suppose that
  {\sc SMM}$(S, k, k')$ performs $\ell$ phases.  It is immediate to
  see that $r_T \leq 4d_\ell$. As was proved in~\cite{CharikarCFM04},
  $ 4d_\ell \leq 8r_{k'}^*$, thus $r_T \leq 8r_{k'}^*$. Consider now
  an optimal clustering of $S$ with $k$ centers and range $r_k^*$ and,
  for notational convenience, define $\epsilon''=\epsilon'/32$. From
  the doubling dimension property, we know that there exist at most
  $k'$ balls in the space (centered at nodes not necessarily in $S$)
  of radius at most $\epsilon'' r_k^*$ which contain all of the points
  in $S$. By choosing one arbitrary center in $S$ for each such ball,
  we obtain a feasible solution to the $k'$-center problem for $S$
  with range at most $2\epsilon'' r_k^*$. Consequently,
  $r_{k'}^* \le 2\epsilon'' r_k^*$. Hence, we have that
  $r_T \le 8r_{k'}^* \le 16\epsilon'' r_k^*$. By
  Fact~\ref{fact:opt-radius-farness}, we know that
  $r_k^* \le \rho_{k}^*$. Therefore, we have
  $r_T \le 16\epsilon'' \rho_k^* = (\epsilon'/2)\rho_k^*$.  Given a
  set $X\subseteq S$ of size $k$, the desired proxy function
  $p(\cdot)$ is the one that maps each point $x\in X$ to the closest
  point in $T$. By the discussion above, we have that
  $d(x, p(x))\le (\epsilon'/2)\rho_k^*$.
\end{proof}

For the diversity problems mentioned in Lemma~\ref{lem:remote-csbt},
we need that for each point of an optimal solution the final core-set
extracted from the data stream contains a \emph{distinct} point very
close to it.  In what follows, we describe a variant of {\sc SMM},
dubbed {\sc SMM-EXT}, which ensures this property. Algorithm {\sc
  SMM-EXT} proceeds as {\sc SMM} but maintains for each $t\in T$ a set
$E_t$ of at most $k$ delegate points close to $t$, including $t$
itself. More precisely, at the beginning of the algorithm, $T$ is
initialized with the first $k'+1$ points of the stream, as before, and
$E_t$ is set equal to $\{t\}$, for each $t \in T$. In the merge step
of Phase $i$, with $i \ge 1$, iteratively for each point $t_1$ not
included in the independent set $I$, we determine an arbitrary point
$t_2 \in I$ such that $d(t_1, t_2)\le 2d_i$ and let $E_{t_2}$ inherit
$\max\{|E_{t_1}|,k-|E_{t_2}|\}$ points of $E_{t_1}$. Note that one
such point $t_2$ must exist, otherwise $I$ would not be a maximal
independent set. Also, note that a point $t_2 \in I$ may inherit
points from sets associated with different points not in $I$. Consider
the update step of Phase $i$ and let $p$ be a new point from the
stream. Let $t \in T$ be the point currently in $T$ which is closest
to $p$. If $d(p, t) > 4d_i$ we add it to $T$. If instead
$d(p,t) \le 4d_i$ and $|E_t| < k$, then we add $p$ to $E_t$, otherwise
we discard it. Finally, we define $T'=\bigcup_{t\in T} E_t$ to be the
output of the algorithm, and observe that $T \subseteq T'$.

\begin{lemma}\label{lem:smm-ext-properties}
  For any $0 < \epsilon' \leq 1$, let $k'=(64/\epsilon')^D\cdot k$,
  and let $T'$ be the set of points returned by {\sc
    SMM-EXT}$(S, k, k')$. Then, given an arbitrary set $X\subseteq S$
  with $|X|=k$, there exist an injective function $p: X\rightarrow T'$
  such that, for any $x \in X$,
  $d(x, p(x)) \le (\epsilon'/2)\rho_k^*$.
\end{lemma}
\begin{proof}
  Let $r_{T'}=\max_{p\in S} d(p, T')$ be the range of $T'$, and
  suppose that {\sc SMM}$(S, k, k')$ performs $\ell$ phases. By
  defining $\epsilon''=\epsilon'/64$, and by reasoning as in the proof
  of Lemma~\ref{lem:smm-properties} we can show that
  $r_{T'} \le 4d_{\ell} \le 16\epsilon''\rho_k^*$. Consider a point
  $x \in X$. If $x \in T'$ then we define $p(x)=x$. Otherwise, suppose
  that $x$ is discarded during Phase~$j$, for some $j$, because either
  in the merging or in the update step the set $E_t$ that was supposed
  to host it had already $k$ points. Let $T_i$ denote the set $T$ at
  the end of Phase~$i$, for any $i \ge 1$. A simple inductive argument
  shows that at the end of each Phase~$i$, with $j \le i \le \ell$
  there is a point $t \in T_i$ such that $|E_t|=k$ and
  $d(x,t) \leq 4d_i$. In particular, there exists a point
  $t \in T_{\ell}$ such that $|E_t|=k$ and
  $d(x,t) \leq 4d_{\ell} \le 16\epsilon''\rho_k^*$.  Since
  $E_t \subset T'$, any point in $E_t$ is at distance at most
  $4d_{\ell} \le 16\epsilon''\rho_k^*$ from $t$, and $|X|=k$, we can
  select a proxy $p(x)$ for $x$ from the $k$ points in $E_t$ such that
  $d(x,p(x)) \le 32\epsilon''\rho_k^* = (\epsilon'/2)\rho_k^*$ and
  $p(x)$ is not a proxy for any other point of $X$.
\end{proof}

It is easy to see that the set $T$ characterized in
Lemma~\ref{lem:smm-properties} satisfies the hypotheses of
Lemma~\ref{lem:remote-edge}. Similarly,
the set $T'$ of Lemma~\ref{lem:smm-ext-properties} satisfies the
hypotheses of Lemma~\ref{lem:remote-csbt}. Therefore, as a consequence
of these lemmas, for metric spaces with bounded doubling dimension
$D$, we have that {\sc SMM} and {\sc SMM-EXT} compute
$(1+\epsilon)$-core-sets for the problems listed in
Table~\ref{tab:diversity-notions}, as stated by the following two
theorems.

\begin{theorem}\label{thm:streaming-remote-edge}
  For any $0 < \epsilon \leq 1$, let $\epsilon'$ be such that
  $(1-\epsilon')=1/(1+\epsilon)$, and let
  $k'=(32/\epsilon')^D\cdot k$. Algorithm {\sc SMM}$(S, k, k')$
  computes a $(1+\epsilon)$-core-set for the remote-edge and
  remote-cycle problems using $\BO{(1/\epsilon)^Dk}$ memory.
\end{theorem}

\begin{theorem}\label{thm:streaming-remote-clique}
For any $0 < \epsilon \leq 1$, let $\epsilon'$ be such that
$(1-\epsilon')=1/(1+\epsilon)$, and let $k'=(64/\epsilon')^D\cdot k$. 
  Algorithm {\sc SMM-EXT}$(S, k, k')$ computes a
  $(1+\epsilon)$-core-set for the remote-clique, remote-star,
  remote-bipartition, and remote-tree problems using
  $\BO{(1/\epsilon)^Dk^2}$ memory.
\end{theorem}
\paragraph{Streaming Algorithm.}  The core-sets discussed above
can be immediately applied to yield the following streaming algorithm
for diversity maximization.  Let $S$ be the input stream of $n$ points.
One pass on the data is performed using  {\sc SMM}, or
{\sc SMM-EXT}, depending on the problem, to compute a core-set 
in main memory. At the end of the pass, a sequential
approximation algorithm is run on the core-set to compute the final solution.
The following theorem is immediate.
\begin{theorem} \label{thm:streaming-1-pass} 
Let $S$ be a stream of $n$ points of
  a metric space of doubling dimension $D$, and let $A$ be a
  linear-space sequential approximation algorithm for any one of the
  problems of Table~\ref{tab:diversity-notions}, returning a solution
  $S'\subseteq S$, with $\diversity_k(S)\leq \alpha\diversity(S')$,
  for some constant $\alpha \geq 1$. Then, for any $0< \epsilon \le 1$,
  there is a 1-pass streaming algorithm for the same problem yielding an
  approximation factor of $\alpha+\epsilon$, with memory
  \begin{itemize}
  \item $\BT{(\alpha/\epsilon)^Dk}$ for the remote-edge
    and the remote-cycle problems;
  \item $\BT{(\alpha/\epsilon)^Dk^2}$ for the remote-clique,
    the remote-star, the remote-bipartition, and the remote-tree
    problems.
  \end{itemize}
\end{theorem}

\section{Applications to MapReduce}
\label{sec:mapreduce}

Recall that a MapReduce (MR) algorithm \cite{KarloffSV10,
  PietracaprinaPRSU12} executes as a sequence of \emph{rounds} where,
in a round, a multiset $X$ of key-value pairs is transformed into a
new multiset $Y$ of pairs by applying a given reducer function (simply
called \emph{reducer}) independently to each subset of pairs of $X$
having the same key. The model features two parameters $M_T$ and
$M_L$, where $M_T$ is the total memory available to the computation,
and $M_L$ is the maximum amount of memory locally available to each
reducer. Typically, we seek MR algorithms that, on an input of size
$n$, work in as few rounds as possible while keeping $M_T=\BO{n}$ and
$M_L=\BO{n^{\delta}}$, for some $0\le {\delta} < 1$.
 
Consider a set $S$ belonging to a metric space of doubling dimension
$D$, and a partition of $S$ into $\ell$ disjoints sets
$S_1, S_2, \ldots, S_{\ell}$ . In what follows, $\diversity(\cdot)$
denotes the diversity function of the problem under consideration, and
$O$ denotes an optimal solution to the problem with respect to
instance $S=\cup_{i=1}^{\ell}S_i$. Also, we let $\rho^*_{k,i}$ be the
optimal farness for $S_i$ w.r.t. $k$, with $1 \le i \le \ell$, and let
$\rho^*_k$ be the optimal farness for $S$ w.r.t. $k$.  Clearly,
$\rho^*_{k,i} \le \rho^*_k$, for every $1 \le i \le \ell$.

The basic idea of our MR algorithms is the following. First, each set
$S_i$ is mapped to a reducer, which computes a core-set
$T_i\subseteq S_i$.  Then, the core-sets are aggregated into one
single core-set $T=\bigcup_{i=1}^\ell T_i$ in one reducer, and a
sequential approximation algorithm is run on $T$, yielding the final
output.  We are thus employing the composable core-sets framework
introduced in~\cite{IndykMMM14}.

The following Lemma shows that if we run Algorithm {\sc GMM} 
from Section~\ref{sec:basic-properties} on
each $S_i$, with $1\le i\le\ell$, and then take the union of the
outputs, the resulting set satisfies the hypotheses of
Lemma~\ref{lem:remote-edge}.

\begin{lemma}\label{lem:gmm-properties}
  For any $0 < \epsilon' \leq 1$, let $k'=(8/\epsilon')^D\cdot k$, and
  let $T=\bigcup_{i=1}^\ell \operatorname{GMM}(S_i, k')$. Then, given
  an arbitrary set $X\subseteq S$ with $|X| = k$, there exist a
  function $p: X \rightarrow T$ such that for any $x\in X$,
  $d(x, p(x)) \le (\epsilon'/2)\rho_k^*$.
\end{lemma}
\begin{proof}
  Fix an arbitrary index $i$, with $1 \leq i \leq \ell$, and let
  $T_i=\{ c_1, c_2, \ldots, c_{k'}\}$, where $c_j$ denotes the point
  added to $T_i$ at the $j$-th iteration of {\sc GMM}$(S_i, k')$. Let
  also $T_i(k)=\{c_1, c_2, \ldots, c_k\}$ and
  $d_k = d(c_k, T_i(k)\setminus \{c_k\})$. From the anticover property
  exhibited by GMM, which holds for any prefix of points selected by
  the algorithm, we have
  $r_{T_i(k)} \le d_k\le \rho_{T_i(k)} \le \rho_k^*$.  Define
  $\epsilon''=\epsilon'/8$. Since $S_i$ can be covered with $k$ balls
  of radius at most $d_k$, and the space has doubling dimension $D$,
  then there exist $k'$ balls in the space (centered at nodes not
  necessarily in $S_i$) of radius at most $\epsilon'' d_k$ that
  contain all the points in $S_i$. By choosing one arbitrary center in
  $S_i$ in each such ball, we obtain a feasible solution to the
  $k'$-center problem for $S_i$ with range at most $2\epsilon'' d_k$,
  which implies that the cost of the optimal solution to $k'$-center
  is at most $2\epsilon' d_k$. As a consequence, {\sc GMM}$(S_i, k')$
  will return a 2-approximate solution $T_i$ to $k'$-center with
  $r_{T_i} \le 4\epsilon'' d_k$, and we have
  $r_{T_i} \le 4\epsilon'' d_k \le 4\epsilon''\rho_k^*$. Let now
  $T=\bigcup_{i=1}^\ell T_i$ and $r_T=\max_{1\le i\le\ell}
  r_{T_i}$. We have that $r_T\le 4\epsilon''\rho_k^*$, hence, for any
  set $X\subseteq S$, the desired proxy function $p(\cdot)$ is
  obtained by mapping each $x\in X$ to the closest point in $T$. By
  the observations on the range of $T$, we have
  $d(x, p(x)) \le 4\epsilon''\rho_k^* = (\epsilon'/2)\rho_k^*$.
\end{proof}

For the diversity problems considered in Lemma~\ref{lem:remote-csbt}
(remote-cycle, remote-star, remote-bipartition, and remote-tree) the
proxy function is required to be injective. Therefore, we develop an
extension of the {\sc GMM} algorithm, dubbed {\sc GMM-EXT} (see
Algorithm~\ref{alg:GMM-EXT} above) which first determines a kernel
$T'$ of $k'\geq k$ points by running {\sc GMM}$(S,k')$ and then
augments $T'$ by first determining the clustering of $S$ whose centers
are the points of $T'$ and then picking from each cluster its center
and up to $k-1$ delegate points.  In this fashion, we ensure that each
point of an optimal solution to the diversity problem under
consideration will have a distinct close ``proxy'' in the returned set
$T$.

\begin{algorithm}[t]
  \caption{{\sc GMM-EXT}($S, k, k'$)}
  \label{alg:GMM-EXT}
  \DontPrintSemicolon 
  \Let{$T'$}{$\text{\sc GMM}(S,k')$}\;
  Let $T' = \{c_1, c_2, \ldots, c_{k'}\}$\;
  \Let{$T$}{$\emptyset$}\;
  \For{\Let{$j$}{$1$} \To $k'$} {
    \Let{$C_j$}{$\{p\in S: c_j=\argmin_{c\in T'} d(c,p) \wedge p \not\in C_h
      \mbox{ with $h < j$}\}$} \;
    \Let{$E_j$}{$\{c_j\}$ $\cup$ $\{$ arbitrary $\min\{|C_j|-1, k-1\}$ points in $C_j\}$ }\;
    \Let{$T$}{$T\cup E_j$}\;
  }
  \Return{$T$}\;
\end{algorithm}

As before, let $S_1,S_2,\dots,S_\ell$ be disjoint subsets of a metric
space of doubling dimension $D$. We have:
\begin{lemma}\label{lem:gmm-ext-properties}
  For any $0 < \epsilon' \leq 1$, let $k'=(16/\epsilon')^d\cdot k$,
  and let $T=\bigcup_{i=1}^\ell \operatorname{GMM-EXT}(S_i, k,
  k')$. Then, given an arbitrary set $X \subseteq S$, with $|X|=k$,
  there exist an injective function $p: X \rightarrow T$ such that for
  any $x\in X$, $d(x, p(x)) \le (\epsilon'/2)\rho_k^*$.
\end{lemma}
\begin{proof}
  For any $1\le i\le\ell$, let
  $T_i=\operatorname{GMM-EXT}(S_i, k ,k')$ be the result of the
  invocation of {\sc GMM-EXT} on $S_i$. By defining
  $\epsilon''=\epsilon'/16$ and by reasoning as in
  Lemma~\ref{lem:gmm-properties}, we have that the range of the set
  $T'_i$ computed by the call to {\sc GMM}$(S_i, k')$ within {\sc
    GMM-EXT}$(S_i,k,k')$ is $r_{T_i'}\le 4\epsilon''\rho_k^*$. Fix an
  arbitrary index $i$, with $1 \leq i \leq \ell$, and consider, for
  $1\leq j\leq k'$, the sets $C_{i,j}$ and $E_{i,j}$ as determined by
  Algorithm {\sc GMM-EXT}$(S_i,k,k')$, and define
  $X_{i,j} = X\cap C_{i,j}$. Since
  $|X_{i,j}|\leq\min\{k,|C_{i,j}|\} = |E_{i,j}|$, we can associate
  each point in $x\in X_{i,j}$ to a distinct proxy $p(x)\in E_{i,j}$.
  Since both $x$ and $p(x)$ belong to $C_{i,j}$, by the triangle
  inequality we have that
  $d(x,p(x))\leq 2r_{T'} \le 8\epsilon''\rho_k^* =
  (\epsilon'/2)\rho_k^*$. Since the input sets
  $S_1, S_2, \dots, S_\ell$ are disjoint, then we have that all the
  $X_{i,j}$ are disjoint. This ensures that we can find a distinct
  proxy for each point of $X$ in $T=\bigcup_{i=1}^\ell T_i$, hence,
  the proxy function is injective.
\end{proof}

The two lemmas above guarantee that the set of points obtained by
invoking {\sc GMM} or {\sc GMM-EXT} on the partitioned input complies
with the hypotheses of Lemmas~\ref{lem:remote-edge}
and~\ref{lem:remote-csbt} of
Section~\ref{sec:basic-properties}.  Therefore, for metric spaces with
bounded doubling dimension $D$, we have that {\sc GMM} and {\sc
  GMM-EXT} compute $(1+\epsilon)$-composable core-sets for the
problems listed in Table~\ref{tab:diversity-notions}, as stated by the
following two theorems.

\begin{theorem}\label{thm:mr-edge-cycle}
  For any $0 < \epsilon \leq 1$, let $\epsilon'$ be such that
  $(1-\epsilon')=1/(1+\epsilon)$, and let $k'=(8/\epsilon')^D\cdot
  k$. The algorithm {\sc GMM}$(S, k')$ computes a $(1+\epsilon)$-composable
  core-set for the remote-edge and remote-cycle problems.
\end{theorem}

\begin{theorem}\label{thm:mr-csbt}
  For any $0 < \epsilon \leq 1$, let $\epsilon'$ be such that
  $(1-\epsilon')=1/(1+\epsilon)$, and let
  $k'=(16/\epsilon')^D\cdot k$. The algorithm {\sc
    GMM-EXT}$(S, k, k')$ computes a $(1+\epsilon)$-composable core-set
  for the remote-clique, remote-star, remote-bipartition, and
  remote-tree problems.
\end{theorem}

\paragraph{MapReduce Algorithm.}  The composable core-sets discussed above
can be immediately applied to yield the following MR algorithm
for diversity maximization.  Let $S$ be the input set of $n$ points
and consider an arbitrary partition of $S$ into $\ell$ subsets $S_1,
S_2, \ldots, S_\ell$, each of size $n/\ell$. In the first round, each
$S_i$ is assigned to a distinct reducer, which computes the
corresponding core-set $T_i$, according to algorithms {\sc GMM}, or
{\sc GMM-EXT}, depending on the problem. In the second round, the
union of the $\ell$ core-sets $T = \bigcup_{i=1}^{\ell}T_i$ is
concentrated within the same reducer, which runs a sequential
approximation algorithm on $T$ to compute the final solution.
We have:
\begin{theorem} \label{thm-2-rounds} Let $S$ be a set of $n$ points of
  a metric space of doubling dimension $D$, and let $A$ be a
  linear-space sequential approximation algorithm for any one of the
  problems of Table~\ref{tab:diversity-notions}, returning a solution
  $S'\subseteq S$, with $\diversity_k(S)\leq \alpha\diversity(S')$,
  for some constant $\alpha \geq 1$. Then, for any
  $0< \epsilon \le 1$, there is a 2-round MR algorithm for the same
  problem yielding an approximation factor of $\alpha+\epsilon$, with
  $M_T=n$ and
  \begin{itemize}
  \item $M_L=\BT{\sqrt{(\alpha/\epsilon)^Dkn}}$ for the remote-edge
    and the remote-cycle problems;
  \item $M_L=\BT{k\sqrt{(\alpha/\epsilon)^Dn}}$ for the remote-tree,
    the remote-clique, the remote-star, and the remote-bipartition
    problems.
  \end{itemize}
\end{theorem}
\begin{proof}
Set $\epsilon'$ such that $1/(1-\epsilon') = 1+ \epsilon/\alpha$, and
recall the the remote-edge and the remote-cycle problems admit
composable core-sets of size $k'=(8/\epsilon')^Dk$, while the problems
remote-tree, remote-clique, remote-star, and remote-bipartition have
core-sets of size $kk'$, with $k'=(16/\epsilon')^D k$.  Suppose that
the above MR algorithm is run with $\ell = \sqrt{n/k'}$ for the former
group of two problems, and $\ell = \sqrt{n/(kk')}$ for the latter
group of four problems.  Observe that by the choice of $\ell$ we have
that both the size of each $S_i$ and the size of the aggregate set
$|T|$ are $O(M_L)$, therefore the stipulated bounds on the local
memory of the reducers are met. The bound on the approximation factor
of the resulting algorithm follows from the fact that the
Theorems~\ref{thm:mr-edge-cycle} and~\ref{thm:mr-csbt} imply that, for
all problems, $\diversity_k(S) \leq
(1+\epsilon/\alpha)\diversity_k(T)$ and the properties of algorithm
$A$ yield $\diversity_k(T)\leq \alpha\diversity(S)$.
\end{proof}

Theorem~\ref{thm-2-rounds} implies that on spaces of constant doubling
dimension, we can get approximations to remote-edge and remote-cycle
in 2 rounds of MR which are almost as good as the best sequential
approximations, with polynomially sublinear local memory
$M_L=\BO{\sqrt{kn}}$, for values of $k$ up to $n^{1-{\delta}}$, while
for the remaining four problems, with polynomially sublinear local
memory $M_L=\BO{k\sqrt{n}}$ for values of $k = \BO{n^{1/2-{\delta}}}$,
for $0\le {\delta} <1$. In fact, for these four latter problems and the
same range of values for $k$, we can obtain substantial memory savings
either by using randomization (in two rounds), or, deterministically
with an extra round (as will be shown in
Section~\ref{sec:gen-mapreduce}). We have:
\begin{theorem}\label{thm-2-rounds-rand}
  For the problems of remote-clique, remote-star, remote-bipartition,
  and remote-tree, we can obtain a randomized 2-round MR algorithm
  with the same approximation guarantees stated in
  Theorem~\ref{thm-2-rounds} holding with high probability, and with
% \[
% M_L = \left\{\begin{array}{ll}
% \BT{\sqrt{(\alpha/\epsilon)^Dkn\log n}} & \mbox{for  } k = \BO{(\epsilon^Dn\log n)^{1/3}} \\
% \BT{({\alpha/\epsilon})^Dk^2} & \mbox{for  }  \BOM{(\epsilon^Dn\log n)^{1/3}} = k =  \BO{n^{1/2-{\delta}}}, 0\leq {\delta}<1/6.\\
% \end{array}\right.
% \]
  \[
    M_L = \left\{
      \begin{aligned}
        &\BT{\sqrt{(\alpha/\epsilon)^Dkn\log n}} \text{ for } k = \BO{(\epsilon^Dn\log n)^{1/3}} \\
        &\BT{({\alpha/\epsilon})^Dk^2} \text{ for }
        k = \left\{
          \begin{aligned}
            &\BOM{(\epsilon^Dn\log n)^{1/3}}\\
            &\BO{n^{1/2-{\delta}}} ~\forall {\delta}\in[0, 1/6)
          \end{aligned}
        \right.
      \end{aligned}
    \right.
  \] 
where $\alpha$ is the approximation guarantee given by the current
best sequential algorithms referenced in
Table~\ref{tab:diversity-notions}.
\end{theorem}
\begin{proof}
  We fix $\epsilon'$ and $k'$ as in the proof of
  Theorem~\ref{thm-2-rounds}, and, at the beginning of the first
  round, we use random keys to partition the $n$ points of $S$ among
  \[
    \ell = \BT{\min\{ \sqrt{n/(k'\log n)} , n/(kk') \}}
  \]
  reducers. Fix
  any of the four problems under consideration and let $O$ be a given
  optimal solution. A simple balls-into-bins argument suffices to show
  that, with high probability, none of the $\ell$ partitions may
  contain more than $\BT{ \max\{\log n, k/\ell\}}$ out of the $k$
  points of $O$.  Therefore, it is sufficient that, within each subset
  of the partition, {\sc GMM-EXT} selects up to those many delegate
  points per cluster (rather than $k-1$).  This suffices to establish
  the new space bounds.
\end{proof}

The deterministic strategy underlying the 2-round MR algorithm can be
employed recursively to yield an algorithm with a larger (yet
constant) number of rounds for the case of smaller local memory
budgets. Specifically, let $T = \bigcup_{i=1}^{\ell}T_i$ be as in the
proof of Theorem~\ref{thm:mr-csbt}.  If $|T|>M_L$, we may re-apply the
core-set-based strategy using $T$ as the new input.  The following
theorem shows that this recursive
strategy can still guarantee an approximation comparable to the
sequential one as long as the local memory $M_L$ is not too small.
\begin{theorem} \label{thm-multi-rounds} Let $S$ be a set of $n$
  points of a metric space of doubling dimension $D$, let and $A$ be a
  linear-space sequential approximation algorithm for any one of the
  problems of Table~\ref{tab:diversity-notions}, returning a solution
  $S'\subseteq S$, with $\diversity_k(S)\leq \alpha\diversity(S')$,
  for some constant $\alpha \geq 1$.  Then, for any
  $0<\epsilon \leq 1$ and $0 < \gamma \leq 1/3$ there is an
  $\BO{(1-\gamma)/\gamma}$-round MR algorithm for the same problem
  yielding an approximation factor of $\alpha+\epsilon$, with $M_T=n$
  and
  \begin{itemize}
  \item $M_L=\BT{(\alpha 2^{(1-\gamma)/\gamma}/\epsilon)^Dkn^\gamma}$ for the remote-edge
    and the remote-cycle problems;
  \item $M_L=\BT{(\alpha 2^{(1-\gamma)/\gamma} \epsilon)^Dk^2n^\gamma}$, for some $\gamma>0$ for the remote-clique,
    the remote-star, the remote-bipartition, and the remote-tree
    problems.
  \end{itemize}
\end{theorem}
\begin{proof}
  Let $\epsilon'$ be such that
  $1/(1-\epsilon') = 1+ \epsilon / (\alpha( 2^{(1-\gamma)/\gamma}
  -1))$ and recall that the the remote-edge and the remote-cycle
  problems admit composable core-sets of size $k'=(8/\epsilon')^Dk$,
  while the problems remote-tree, remote-clique, remote-star, and remote-bipartition,
  have core-sets of size $kk'$, with
  $k'=(16/\epsilon')^D$. We may apply the following recursive
  strategy. We partition the input set $S$ into $n/M_L$ sets of size
  $M_L$ and compute the corresponding core-sets. Let $T$ be the union
  of these core-sets. If $|T| > M_L$, then we recursively apply the
  same strategy using $T$ as the new input set, otherwise, we send $T$
  to a single reducer where algorithm $A$ is applied. By the choice of
  the parameters, it follows that in all cases $(1-\gamma)/\gamma$
  rounds suffice to shrink the input set to size at most $M_L$. The
  resulting approximation factor with respect to $\diversity_k(S)$
  will then be at most
  \[
    \begin{aligned}
      \alpha \left(
        1+{\epsilon \over \alpha (2^{(1-\gamma)/\gamma} -1)}
      \right)^{\frac{(1-\gamma)}{\gamma}}
      \le 
      \alpha \left(1+ 
        \frac
        {\epsilon (2^{(1-\gamma)/\gamma} -1)}
        {\alpha (2^{(1-\gamma)/\gamma} -1)} \right)
      = \alpha + \epsilon,
    \end{aligned}
  \] 
  where the last inequality follows from the known fact
  $(1+a)^b \leq (1+(2^b-1)a)$ for every $a \in [0,1]$ and $b > 1$, and
  the observation that, by the choice of $\gamma$, we have
  $(1-\gamma)/\gamma \geq 2$.
\end{proof}

\section{Saving memory: generalized \mbox{core-sets}} \label{sec:generalized}

Consider the problems remote-clique, remote-star, remote-bipartition,
and remote-tree. Our core-sets for these problems are obtained by
exploiting the sufficient conditions stated in
Lemma~\ref{lem:remote-csbt}, which require the existence of an
injective proxy function that maps the points of an optimal solution
into close points of the core-set.  To ensure this property, our
strategy so far has been to add more points to the core-sets.  More
precisely, the core-set is composed by a kernel of $k'$ points,
augmented by selecting, for each kernel point, a number of up to $k-1$
delegate points laying within a small range.  This augmentation
ensures that for each point $o$ of an optimal solution $O$, there
exists a distinct close proxy among the delegates of the kernel point
closest to $o$, as required by Lemma~\ref{lem:remote-csbt}.

In order to reduce the core-set size, the augmentation can be done
implicitly by keeping track only of the number of delegates that must
be added for each kernel point. A set of pairs $(p, m_p)$ is then
returned, where $p$ is a kernel point and $m_p$ is the number of
delegates for $p$ (including $p$ itself).
The intuition behind this approach is the following. The set of pairs
described above can be viewed as a compact representation of a
multiset, where each point $p$ of the kernel appears with multiplicity
$m_p$. If, for a given diversity measure, we solve the natural
generalization of the maximization problem on the multiset, then we
can transform the obtained multiset solution into a feasible solution
for $S$ by selecting, for each multiple occurrence of a kernel point,
a distinct close enough point in $S$.  In what follows we illustrate
this idea in more detail.

Let $S$ be a set of points. A \emph{generalized core-set} $T$ for $S$
is a set of pairs $(p,m_p)$ with $p \in S$ and $m_p$ a positive
integer, referred to as the \emph{multiplicity} of $p$, where the
first components of the pairs are all distinct. We define its
\emph{size} $s(T)$ to be the number of pairs it contains, and its
\emph{expanded size} as $m(T) = \sum_{(p,m_p) \in T} m_p$.
Moreover, we define the \emph{expansion} of a generalized core-set $T$
as the multiset $\mathcal{T}$ formed by including, for each pair
$(p, m_p)\in T$, $m_p$ replicas of $p$ in $\mathcal{T}$.

Given two generalized core-sets $T_1$ and $T_2$, we say that $T_1$ is
a \emph{coherent subset} of $T_2$, and write $T_1 \sqsubseteq T_2$, if
for every pair $(p,m_p) \in T_1$ there exists a pair
$(p,m'_p) \in T_2$ with $m'_p \geq m_p$. For a given diversity
function $\diversity$ and a generalized core-set $T$ for $S$, we
define the \emph{generalized diversity} of $T$, denoted by
$\gendiv(T)$, to be the value of $\diversity$ when applied to its
expansion $\mathcal{T}$, where $m_p$ replicas of the same point $p$
are viewed as $m_p$ distinct points at distance 0 from one another.
We also define the \emph{generalized $k$-diversity} of $T$ as
\[
\gendiv_k(T) = \max_{T' \sqsubseteq T : m(T')=k} \gendiv(T').
\]
Let $T$ be a generalized core-set for a set of points $S$. A set
$I(T) \subseteq S$ with $|I(T)|=m(T)$ is referred to as a
\emph{${\delta}$-instantiation} of $T$ if for each pair $(p,m_p) \in T$
it contains $m_p$ distinct delegate points (including $p$), each at
distance at most ${\delta}$ from $p$, with the requirement that the sets
of delegates associated with any two pairs in $T$ are disjoint. The
following lemma
ensures that the difference between the generalized
diversity of $T$ and the diversity of any of its
${\delta}$-instantiations is bounded.
 
\begin{lemma}\label{lem-gendiv}
  ~Let $T$ be a generalized core-set for $S$ with $m(T)=k$, and
  consider the remote-clique, remote-star, remote-bipartition, and
  remote-tree problems.  For any ${\delta}$-instantiation $I(T)$ of $T$
  we have that
  \[
    \diversity(I(T)) \geq \gendiv(T) - f(k)2{\delta}.
  \] 
  where $f(k) = {k \choose 2}$ for remote-clique, $f(k) = k-1$ for
  remote-star and remote tree, and
  $f(k) = \lfloor k/2 \rfloor \cdot \lceil k/2 \rceil$ for
  remote-bipartition.
\end{lemma}
\begin{proof}
 Recall that $\gendiv(T)$ is defined over the expansion $\mathcal{T}$
 of $T$ where each pair $(p,m_p) \in T$ is represented by $m_p$
 occurrences of $p$.  We create a 1-1 correspondence between
 $\mathcal{T}$ and $I(T)$ by mapping each occurrence of a point
 $p \in \mathcal{T}$ into a distinct proxy chosen among the delegates
 for $(p,m_p)$ in $I(T)$. The lemma follows by noting both
 $\gendiv(T)$ and $\diversity(I(T))$ are expressed in terms of sums
 of $f(k)$ distances and that, by the triangle inequality, for any
 two points $p_1,p_2$ in the multiset (possibly two occurrences of
 the same point $p$) the distance of the corresponding proxies is at
 least $d(p_1,p_2) - 2{\delta}$.
\end{proof}

It is important to observe that the best sequential approximation
algorithms for the remote-clique, remote-star, remote-bipartition, and
remote-tree problems (see Table~\ref{tab:diversity-notions}), which
are essentially based on either finding a maximal matching or running
{\sc GMM} on the input set
\cite{HassinRT97,ChandraH01,HalldorssonIKT99}, can be easily adapted
to work on inputs with multiplicities. We have:
\begin{fact}\label{fact:gen}
  The best existing sequential approximation algorithms for the
  remote-clique, remote-star, remote-bipartition, and remote-tree, can
  be adapted to obtain from a given generalized core-set $T$ a
  coherent subset $\hat{T}$ with expanded size $m(\hat{T})=k$ and
  $\gendiv(\hat{T}) \ge (1/\alpha) \gendiv_k(T)$, where $\alpha$ is
  the same approximation ratio achieved on the original problems.
  The adaptation works in space $\BO{s(T)}$.
\end{fact}
\subsection{Streaming}
Using generalized core-sets we can lower the memory requirements for
the  remote-tree, remote-clique, remote-star, and remote-bipartition
problems to match the one of the other two problems, at the expense of
an extra pass on the data. We have:
\begin{theorem}\label{thm-2-pass-gen}
  For the problems of remote-clique, remote star, remote-bipartition,
  and remote-tree, we can obtain a 2-pass streaming algorithm with
  approximation factor $\alpha+\epsilon$ and memory
  $\BT{(\alpha^2/\epsilon)^Dk}$, for any $0< \epsilon <1$, where
  $\alpha$ is the approximation guarantee given by the current best
  sequential algorithms referenced in
  Table~\ref{tab:diversity-notions}.
\end{theorem}
\begin{proof}
Let $\bar{\epsilon}$ be such that $\alpha+\epsilon =
\alpha/(1-\bar{\epsilon})$, and observe that $\bar{\epsilon} =
\BT{\epsilon/\alpha}$. In the first pass we determine a generalized
core-set $T$ of size $k'=(64\alpha/\bar{\epsilon})^D\cdot k$ by
suitably adapting the {\sc SMM-EXT} algorithm to maintain counts
rather than delegates for each kernel point. Let $r_T$ denote the
maximum distance of a point of $S$ from the closest point $x$ such
that $(x,m_x)$ is in $T$.  Using the argument in the proof of
Lemma~\ref{lem:smm-properties}, setting $\epsilon' = \bar{\epsilon}/(2
\alpha)$, it is easily shown that $r_T \leq (\epsilon'/2) \rho_k^* =
(\bar{\epsilon}/(4\alpha)) \rho_k^*$. Therefore, we can establish an
injective map $p(\cdot)$ from $O$ to the expansion $\mathcal{T}$ of
$T$. Let us focus on the remote-clique problem (the argument for the
other three problems is virtually identical), and define $\bar{\rho} =
\diversity(O) /{k\choose 2}$. By reasoning as in the proof of
Lemma~\ref{lem:remote-csbt}, we can show that $\gendiv_k(T) \geq
\diversity(O)(1-\bar{\epsilon}/(2\alpha))$.

At the end of the pass, the best sequential algorithm for the problem,
adapted as stated in Fact~\ref{fact:gen}, is used to compute in memory
a coherent subset $\hat{T} \sqsubseteq T$ with $m(\hat{T})= k$ and
such that
$\gendiv(\hat{T}) \geq
\diversity(O)(1-\bar{\epsilon}/(2\alpha))/\alpha$.  The second pass
starts with $\hat{T}$ in memory and computes an $r_T$-instantiation
$I(\hat{T})$ by selecting, for each pair $(p,m_p) \in \hat{T}$, $m_p$
distinct delegates at distance at most
$r_T \leq (\bar{\epsilon}/(4\alpha)) \bar{\rho}$ from $p$.  Note that
a point from the data stream could be a feasible delegate for multiple
pairs. Such a point must be retained as long as the appropriate
delegate count for each such pair has not been met.  By applying
Lemma~\ref{lem-gendiv} with $\delta=(\bar{\epsilon}/(4\alpha)) \bar{\rho}$,
we get $\diversity(I(\hat{T})) \geq \diversity(O)/(\alpha+\epsilon)$.
Since $\bar{\epsilon} = \BT{\epsilon/\alpha}$, the space required is
$\BT{(\alpha/\bar{\epsilon})^D k} =
\BT{(\alpha^2/\epsilon)^D k}$.
\end{proof}

\begin{table}
  \small
  \centering
  \begin{tabular}{l@{\hskip 2pt} | c@{\hskip 1pt} c@{\hskip 1pt} | c@{\hskip -1pt} c@{\hskip 0pt} c}
    \toprule
    Problem
     & \multicolumn{2}{c|}{Streaming}
     & \multicolumn{3}{c}{MapReduce}
    \\
    % Header 2
     & 1 pass
     & 2 passes
     & 2 rounds det.
     & 2 rounds randomized
     & 3 rounds det.
    \\
    \midrule
    r-edge 
     & \multirow{2}{*}{$\BT{(1/\epsilon)^D k}$}
     & \multirow{2}{*}{$-$}
     & \multirow{2}{*}{$\BT{\sqrt{(1/\epsilon)^D kn}}$}
     & \multirow{2}{*}{$-$}
     & \multirow{2}{*}{$-$}
    \\
    r-cycle & & & & & \\
    \midrule
    r-clique 
     & \multirow{4}{*}{
       $\BT{(1/\epsilon)^D k^2}$
       }
     & \multirow{4}{*}{
       $\BT{(1/\epsilon)^D k}$
       }
     & \multirow{4}{*}{
       $\BT{k\sqrt{(1/\epsilon)^Dn}}$
       }
     & \multirow{4}{*}{
       $\begin{aligned}&\max\Big\{\BT{({1/\epsilon})^Dk^2},\\ &\qquad\BT{\sqrt{(1/\epsilon)^Dkn\log n}}\Big\}\end{aligned}$
       }
     & \multirow{4}{*}{
       $\BT{\sqrt{(1/\epsilon)^D kn}}$
       }
    \\
    r-star & & & & & \\
    r-bipartition & & & & & \\
    r-tree & & & & & \\
    \bottomrule
  \end{tabular}
  \caption{%
    Memory requirements of our streaming and MapReduce approximation
    algorithms. (For MapReduce we report only the size of $M_L$ since
    $M_T$ is always linear in $n$.)  %
    The approximation factor of each algorithm is $\alpha+\epsilon$,
    where $\alpha$ is the constant approximation factor of the
    sequential algorithms listed in Table~\ref{tab:diversity-notions}.
    }\label{tab:mapreduce-streaming}
\end{table}

\subsection{MapReduce}
\label{sec:gen-mapreduce}
Let $\diversity$ be a diversity function, $k$ be a positive
integer, and $\beta \ge 1$. A function $c(S)$ that maps a set of
points $S$ to a generalized core-set $T$ for $S$ computes a
\emph{$\beta$-composable generalized core-set} for $\diversity$ if,
for any collection of disjoint sets $S_1,\dots, S_\ell$, we have that
\[
  \gendiv_k\left(\bigcup_{i=1}^\ell c(S_i)\right) \ge
  \frac{1}{\beta} \diversity_k\left(\bigcup_{i=1}^\ell S_i\right).
\]
Consider a simple variant of {\sc GMM-EXT}, which we refer to as {\sc
  GMM-GEN}, which on input $S$, $k$ and $k'$ returns a generalized
core-set $T$ of $S$ of size $s(T)=k'$ and extended size
$m(T) \leq k k'$ as follows: for each point $c_i$ of the kernel set
$T' =$ {\sc GMM}$(S,k')$, algorithm {\sc GMM-GEN} returns a pair
$(c_i,m_{c_i})$ where $m_{c_i}$ is equal to the size of the set $E_i$
computed in the $i$-th iteration of the for loop of {\sc GMM-EXT}.

\begin{lemma}\label{lem:gmm-gen-core-set}
For any $\epsilon' >0$, define $k'=(16\alpha/\epsilon')^D k$.
Algorithm {\sc GMM-GEN} computes a $\beta$-composable generalized
core-set for the remote-clique, remote-star, remote-bipartition, and
remote-tree problems, with $1/\beta = 1 - \epsilon'/(2\alpha)$.
\end{lemma}
\begin{proof}
Given a collection of disjoint sets $S_1,\dots, S_\ell$, let $T_i =$
{\sc GMM-GEN}$(S_i,k,k')$, and $T=\bigcup_{i=1}^{\ell} T_i$.  Consider
the expansion $\mathcal{T}$ of $T$.  Let us focus on the remote-clique
problem (the argument for the other three problems is virtually
identical) and define $\bar{\rho} = \diversity(O) /{k\choose 2}$.  By
reasoning along the lines of the proof of
Theorem~\ref{thm-2-pass-gen}, we can establish an injective map $p:
O\rightarrow \mathcal{T}$ such that, for any $o\in O$, $d(o, p(o))\le
(\epsilon'/(4\alpha))\bar{\rho}$.  Let $\hat{T}$ be the generalized
core-set whose expansion into a multiset yields the $k$ points of the
image of $p$. We have:
\[
 \gendiv_k(T) 
 \ge \gendiv(\hat{T})
 \ge \diversity(O)\left(1-\frac{\epsilon'}{2\alpha}\right)
 \qed
\]
\end{proof}
We are now able to show that {\sc GMM-GEN}
computes a high-quality $\beta$-composable generalized core-set, which
can then be employed in a 3-round MR algorithm to approximate the
solution to the four problems under consideration with lower memory
requirements.
This result is summarized in the following theorem.
\begin{theorem}\label{thm:3-rounds-gen}
  For the problems of remote-clique, remote-star, remote-bipartition,
  and remote-tree, we can obtain a 3-round MR algorithm with
  approximation factor $\alpha+\epsilon$ and
  $M_L=\BT{\sqrt{(\alpha^2/\epsilon)^Dkn}}$, for any $0< \epsilon <1$,
  where $\alpha$ is the approximation guarantee given by the current
  best sequential algorithms referenced in
  Table~\ref{tab:diversity-notions}.
\end{theorem}
\begin{proof}
  Consider the remote-clique problem (the argument for the other three
  problems is virtually identical) and define
  $\bar{\rho}=\diversity(O)/{k\choose 2}$. Let $\epsilon'$ be such
  that $\alpha+\epsilon = \alpha/(1-\epsilon')$ and observe that
  $\epsilon'= \BT{\epsilon/\alpha}$. Also, set
  $k'=(16\alpha/\epsilon')^D\cdot k$.  For $\ell = \sqrt{n/k'}$
  consider a arbitrary partition of the input set $S$ into $\ell$
  subsets $S_1, S_2, \ldots, S_{\ell}$ each of size
  $M_L = n/\ell = \sqrt{nk'}$ each. In the first round, each reducer
  applies {\sc GMM-GEN} to a distinct subset $S_i$ to compute
  generalized core-sets of size $k'$. In the second round, these
  generalized core-sets are aggregated in a single generalized
  core-set $T$, whose size is $\ell k' = \sqrt{nk'} = M_L$ and such
  that the maximum distance of a point of $S$ from the closest point
  $x$ with $(x, m_x)\in T$ is
  $r_T\leq (\epsilon'/(4\alpha)) \bar{\rho}$. Then, one reducer
  applies to $T$ the best sequential algorithm for the problem,
  adapted as stated in Fact~\ref{fact:gen}, to compute a coherent
  subset $\hat{T} \sqsubseteq T$ with $m(\hat{T})= k$ and such that
\[
\gendiv(\hat{T}) \geq {1\over \alpha}\gendiv_k(T)
\geq
\left(1-\frac{\epsilon'}{2\alpha}\right){1\over \alpha}\diversity(O),
\]
where the last inequality follows by Lemma~\ref{lem:gmm-gen-core-set}.
In the third round, $\hat{T}$ is distributed to $\ell$ reducers which
are able to compute an instantiation $I(\hat{T})$ of $\hat{T}$ as
follows. For each pair $(p,m_p) \in \hat{T}$, such that $p \in S_i$,
the $i$-th reducer selects $m_p$ distinct delegates from $S_i$ at
distance at most $r_T \leq (\epsilon'/(4\alpha)) \bar{\rho}$ from
$p$. By Lemma~\ref{lem-gendiv}, we have that
\[
\begin{aligned}
\diversity(I(\hat{T})) 
&\ge 
\left(1-\frac{\epsilon'}{2\alpha}\right)\frac{1}{\alpha}\diversity(O)
- \frac{\epsilon'}{2\alpha}\diversity(O)
=
\frac{1}{\alpha}
\left(1 - \frac{\epsilon'}{2\alpha} - \frac{\epsilon'}{2}
     \right)\diversity(O) \\
     &\ge \frac{1}{\alpha}(1-\epsilon')\diversity(O)
     = \frac{1}{\alpha+\epsilon}\diversity(O)
\end{aligned}
\]
As for the memory bound, we have that
$M_L = \sqrt{n k'} = \BT{\sqrt{(\alpha^2/\epsilon)^Dkn}}$.
\end{proof}

A synopsis of the main theoretical results presented in the paper is
given in Table~\ref{tab:mapreduce-streaming}.

\section{Experimental evaluation}
\label{sec:experiments}
We ran extensive experiments on a cluster of 16 machines, each
equipped with 18GB of RAM and an Intel I7 processor.
To the best of our knowledge, ours is the first work on diversity
maximization in the MapReduce and Streaming settings, which
complements theoretical findings with an experimental evaluation.
The MapReduce algorithm has been implemented within the Spark
framework, whereas the streaming algorithm has been implemented in
Scala, simulating a Streaming setting\footnote{The code is available
  as free software at \url{https://github.com/Cecca/diversity-maximization}}.
Since optimal solutions are out of reach for the input sizes that we
considered, for each dataset we computed approximation ratios with
respect to the best solution found by many runs of our MapReduce
algorithm with maximum parallelism and large local memory.
We run our experiments on both synthetic and real-world datasets.
Synthetic datasets are generated randomly from the three-dimensional
Euclidean space in the following way.  For a given $k$, $k$ points are
randomly picked on the surface of the unit radius sphere centered at
the origin of the space, so to ensure the existence of a set of
far-away points, and the other points are chosen uniformly at random
in the concentric sphere of radius 0.8. Among all the distributions
used to test our algorithms, on which we do not report for brevity, we found that this is the most challenging, hence the more
interesting to demonstrate.
To test our algorithm on real-world workloads we used the
\emph{musiXmatch} dataset~\cite{BertinM11}.  This dataset contains
the lyrics of 237,662 songs, each represented by the vector of word
counts of the most frequent 5,000 words across the entire dataset.  The
dimensionality of the space of these vectors is therefore 5,000.  We
filter out songs represented by less than 10 frequent words, obtaining
a dataset of 234,363 songs.  The reason of this filtering is that one
can build an optimal solution using songs with short, non overlapping
word lists. Thus, removing these songs makes the dataset more challenging
for our algorithm.  On this dataset, as a distance between two vectors
$\vec{u}$ and $\vec{v}$, we use the \emph{cosine distance}, defined as
$\dist(\vec{u}, \vec{v}) = \arccos\big(\frac{\vec{u} \cdot
  \vec{v}}{\|\vec{u}\| \|\vec{v}\|}\big)$. This distance is closely
related to the \emph{cosine similarity} commonly used in Information
Retrieval \cite{LeskovecRU14}.
For brevity, we will report the results only for the remote-edge
problem.  We observed similar behaviors for the other diversity
measures, which are all implemented in our software.  All results
reported in this section are obtained as averages over at least 10
runs.

\subsection{Streaming algorithm}\label{sec:streaming-experiments}

The first set of experiments investigates the behavior of the
streaming algorithm for various values of $k$, as well as the impact
of the core-set size, as controlled by the parameter $k'$, on the
approximation quality. The results of these experiments are reported
in Figure~\ref{fig:streaming-approx-musixmatch},
for the musiXmatch dataset, and Figure~\ref{fig:streaming-approx-synthetic}.
for a synthetic dataset of 100 million points, generated as explained above.

\begin{figure}[t]
  \begin{minipage}{.49\linewidth}
    \centering
    \includegraphics[width=\columnwidth]{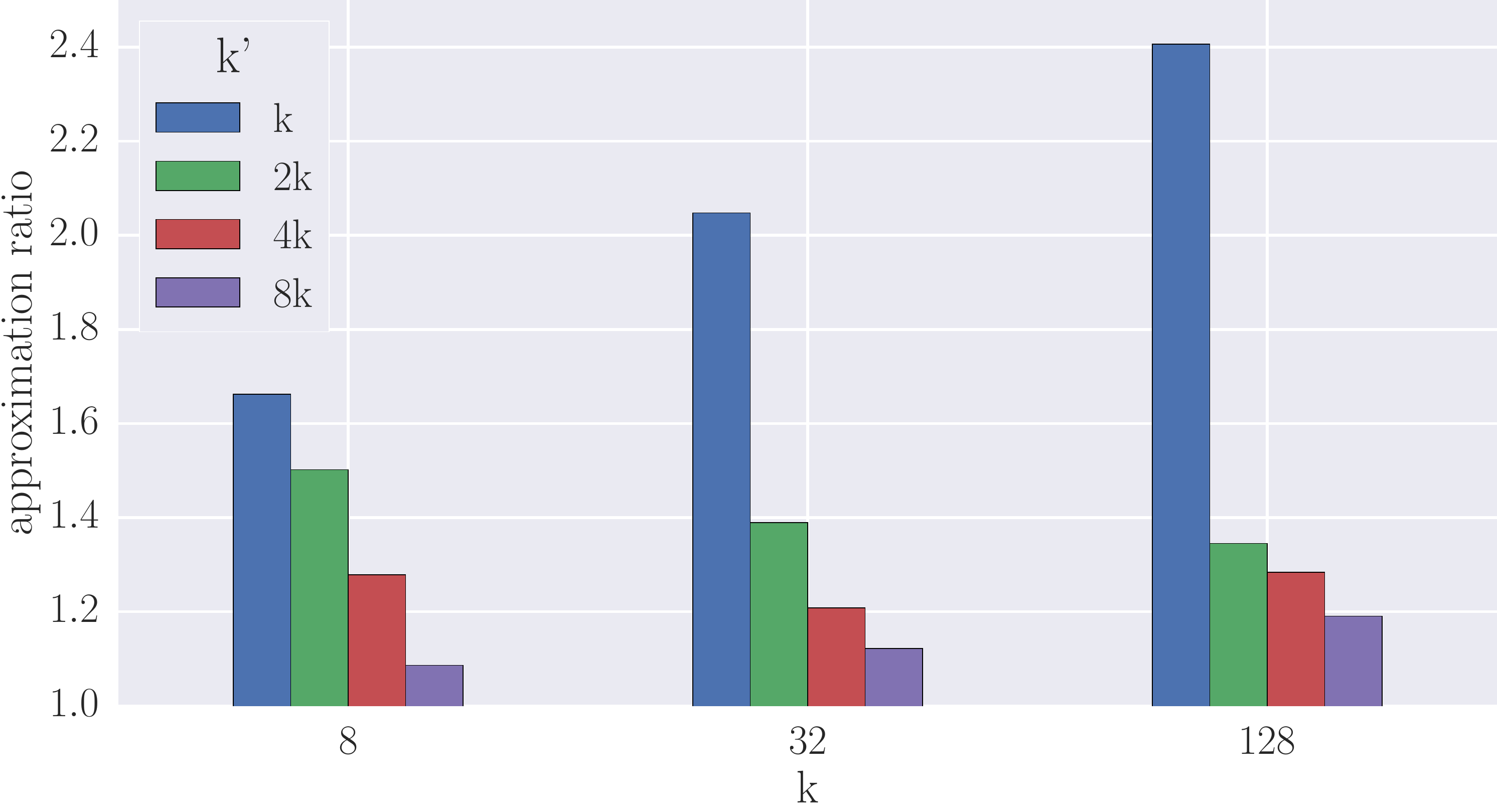}
    \caption{Approximation ratio for the streaming algorithm for different values of $k$ and $k'$ on the \emph{musiXmatch} dataset.}\label{fig:streaming-approx-musixmatch}
  \end{minipage}
  \hfill
  \begin{minipage}{.49\linewidth}
    \centering
    \includegraphics[width=\columnwidth]{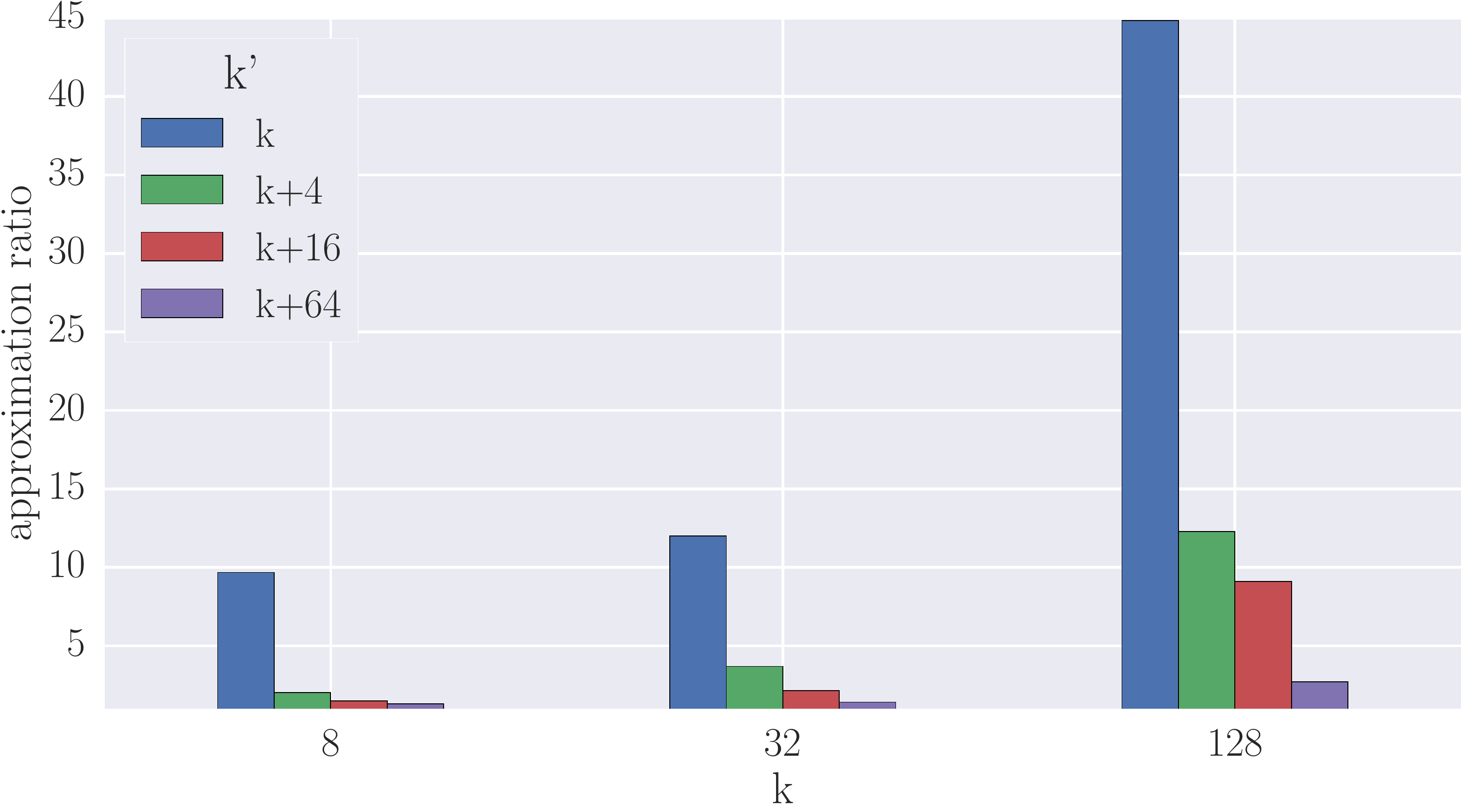}
    \caption{Approximation ratios for the streaming algorithm for different values of $k$ and $k'$ on a synthetic dataset of 100 million points.}\label{fig:streaming-approx-synthetic}
  \end{minipage}
\end{figure}

First, we observe that as $k$ increases the remote-edge measure becomes harder to approximate: finding a higher number of diverse elements is more difficult.
On the real-world dataset, because of the high dimensionality of its
space, we test the influence of $k'$ on the approximation with a
geometric progression of $k'$
(Figure~\ref{fig:streaming-approx-musixmatch}).  On the synthetic
datasets instead (Figure~\ref{fig:streaming-approx-synthetic}), since
$\mathbb{R}^3$ has a smaller doubling dimension, the effect of $k'$ is
evident already with small values, therefore we use a linear
progression.
As expected, by increasing $k'$ the accuracy of the algorithm increases in both
datasets. Observe that although the theory suggests that good
approximations require rather large values of $k' =
\Omega(k/\epsilon^D)$, in practice our experiments show that
relatively small values of $k'$, not much larger than $k$, already
yield very good approximations, even for the real-world dataset
whose doubling dimension is unknown.

In Figure~\ref{fig:throughput}, we consider the performance of the
kernel of streaming algorithm, that is, we concentrate on the time
taken by the algorithm to process each point, ignoring the cost of
streaming data from memory. The rationale is that data may be streamed
from sources with very different throughput: our goal is to show the
maximum rate that can be sustained by our algorithm independently of
the source of the stream.  We report results for the same combination
of parameters shown in Figure~\ref{fig:streaming-approx-musixmatch}.
As expected, the throughput is inversely proportional to both $k$ and
$k'$, with values ranging from 3,078 to 544,920 points/s.  The
throughput supported by our algorithm makes it amenable to be used in
streaming pipelines: for instance, in 2013
Twitter\footnote{\url{https://blog.twitter.com/2013/new-tweets-per-second-record-and-how}}
averaged at 5,700 tweets/s and peaked at 143,199 tweets/s.  In this
scenario, it is likely that the bottleneck of the pipeline would be
the data acquisition rather than our core-set construction.

As for the synthetic dataset, the throughput of the algorithm exhibits
a behavior with respect to $k$ and $k'$ similar to the one reported in
Figure~\ref{fig:throughput}, but with higher values ranging from
78,260 to 850,615 points/s since the distance function is cheaper to compute.

\begin{figure}[t]
  \begin{minipage}{.49\linewidth}
    \centering
    \includegraphics[width=\columnwidth]{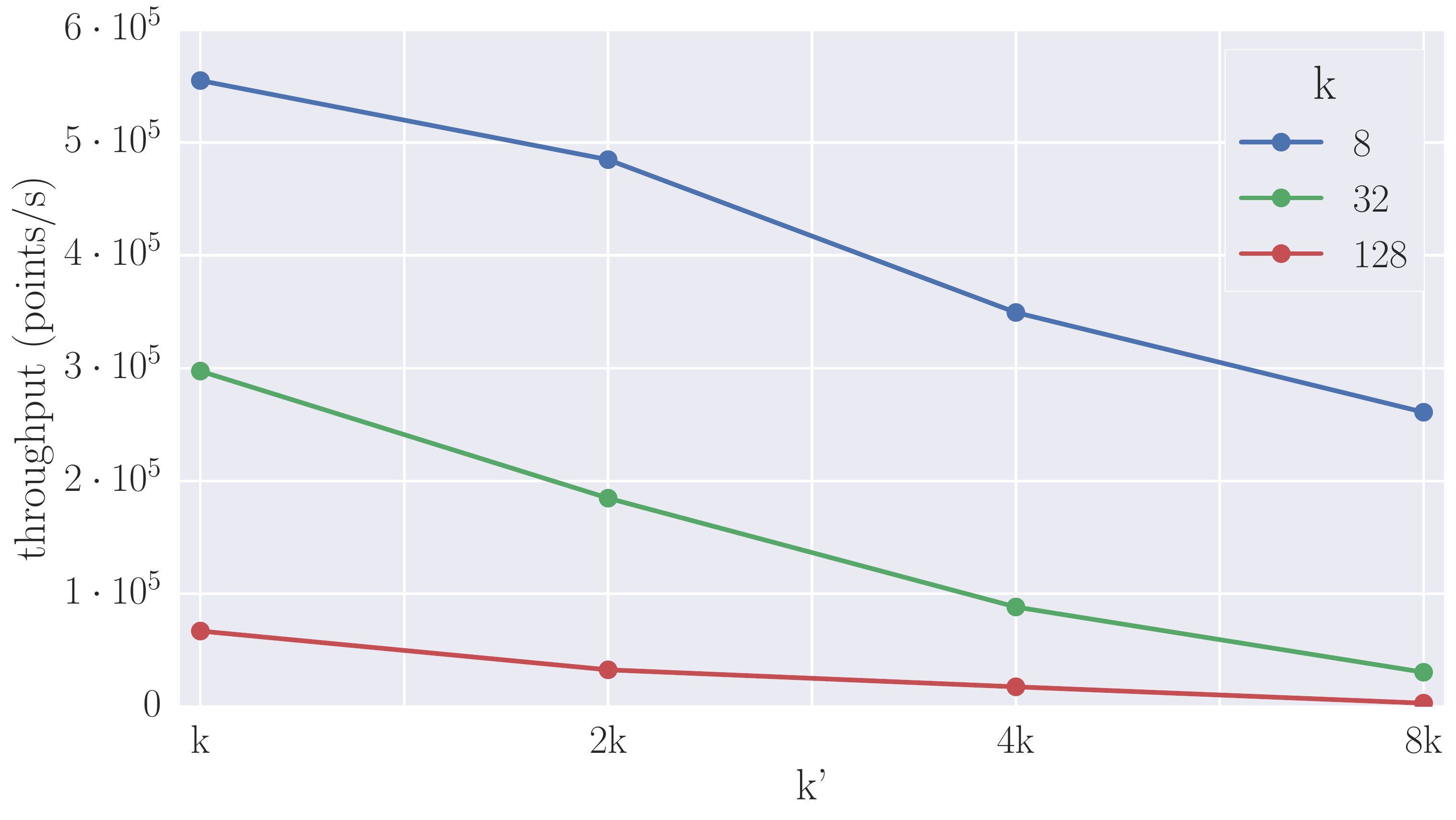}
    \caption{Throughput of the kernel of the streaming algorithm on the \emph{musiXmatch} dataset.}
    \label{fig:throughput}
  \end{minipage}
  \hfill
  \begin{minipage}{.49\linewidth}
    \centering
    \includegraphics[width=\columnwidth]{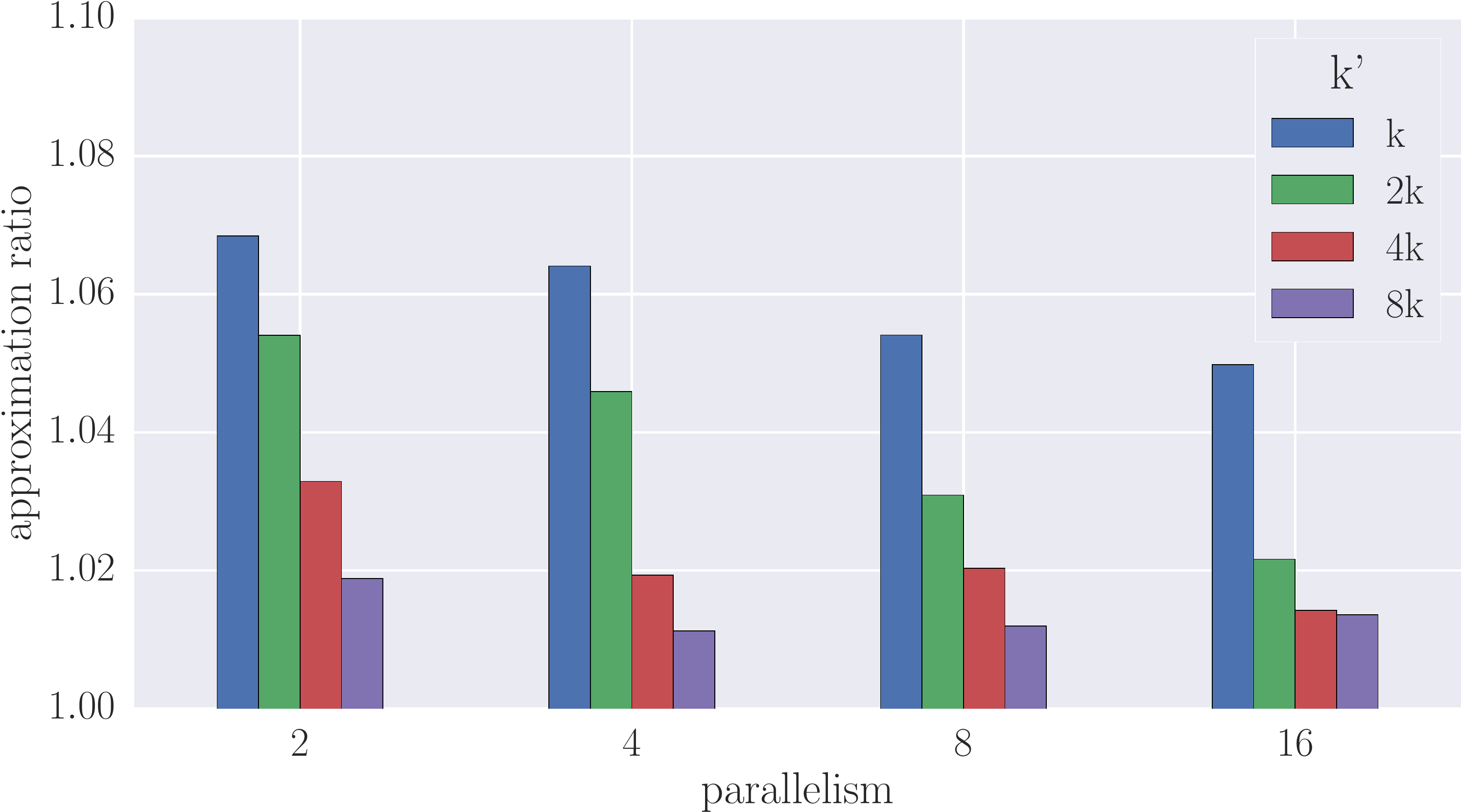}
    \caption{Approximation ratios for the MR algorithm for different values of $k$ and $k'$ on a synthetic dataset of 100 million points.%
    }\label{fig:mapreduce-approximation}
  \end{minipage}
\end{figure}

\subsection{MapReduce algorithm} \label{subsec:MRexp}

We demonstrate our MapReduce algorithm on the same datasets used in
the previous section. For this set of experiments we fixed $k=128$ and
we varied two parameters: size of the core-sets, as controlled by
$k'$, and parallelism (i.e., the number of reducers).  Because the
solution returned by the MapReduce algorithm for $k'=k$ turns out to
be already very good, we use a geometric progression for $k'$ to
highlight the dependency of the approximation factor on $k'$.  The
results are reported in Figure~\ref{fig:mapreduce-approximation}.  For
a fixed level of parallelism, we observe that the approximation ratio
decreases as $k'$ increases, in accordance to the theory.  Moreover,
we observe that the approximation ratios are in general better than
the ones attained by the streaming algorithm, plausibly because in
MapReduce we use a 2-approximation $k'$-center algorithm to build the
core-sets, while in Streaming only a weaker 8-approximation
$k'$-center algorithm is available.

Figure~\ref{fig:mapreduce-approximation} also reveals that if we fix
$k'$ and increase the level of parallelism, the approximation ratio
tends to decrease.  Indeed, the final core-set obtained by aggregating
the ones produced by the individual reducers grows larger as the
parallelism increases, thus containing more information on the input
set.  Instead, if we fix the product of $k'$ and the level of
parallelism, hence the size of the aggregate core-set, we observe that
increasing the parallelism is mildly detrimental to the approximation
quality.  This is to be expected, since with a fixed space budget in
the second round, in the first round each reducer is forced to build a
smaller and less accurate core-set as the parallelism increases.

The experiments for the real-world \emph{musiXmatch} dataset (figures
omitted for brevity) highlight that the {\sc GMM} $k'$-center
algorithm returns very good core-sets on this high dimensional
dataset, yielding approximation ratios very close to 1 even for low
values of $k'$. As remarked above, the more pronounced dependence on
$k'$ in the streaming case may be the result of the weaker approximation
guarantees of its core-set construction.

Since in real scenarios the input might not be distributed randomly
among the reducers, we also experimented with an ``adversarial''
partitioning of the input: each reducer was given points coming from a
region of small volume, so to obfuscate a global view of the
pointset. With such adversarial partitioning, the approximation ratios
worsen by up to $10\%$. On the other hand, as $k'$ increases, the
time required by a random shuffle of the points among the reducers
becomes negligible with respect to the overall running time. Thus,
randomly shuffling the points at the beginning may prove
cost-effective if larger values of $k'$ are affordable.

\subsection{Comparison with state of the art}

In Table~\ref{tab:comparison}, we compare our MapReduce algorithm
(dubbed {\tt CPPU}) against its state of the art competitor presented
in~\cite{AghamolaeiFZ15} (dubbed {\tt AFZ}). Since no code was
available for {\tt AFZ}, we implemented it in MapReduce with the same
optimizations used for {\tt CPPU}.  We remark that {\tt AFZ} employs
different core-set constructions for the various diversity measures,
whereas our algorithm uses the same construction for all diversity
measures.  In particular, for remote-edge, {\tt AFZ} is equivalent to
{\tt CPPU} with $k'=k$, hence the comparison 
is less interesting and can be derived from the behavior
of {\tt CPPU} itself. Instead, for remote-clique, the core-set
construction used by {\tt AFZ} is based on local search and may
exhibit highly superlinear complexity.  For remote-clique, we
performed the comparison with various values of $k$, on datasets of 4
million points on the 2-dimensional Euclidean space, using 16 reducers
({\tt AFZ} was prohibitively slow for higher dimensions and bigger
datasets).  The datasets were generated as described in the
introduction to the experimental section. Also, we ran {\tt CPPU} with
$k'=128$ in all cases, so to ensure a good approximation ratio at the
expense of a slight increase of the running time.
As Table~\ref{tab:comparison} shows, {\tt CPPU} is in all cases at least
three orders of magnitude faster than {\tt AFZ}, while achieving a better quality at the same time.

\begin{table}[t]
  \begin{minipage}{.43\linewidth}
    \centering
    \begin{tabular}{lrrrr}
      \toprule
      & \multicolumn{2}{c}{approximation} & \multicolumn{2}{c}{time (s)} \\
      k & {\tt AFZ} & {\tt CPPU} & {\tt AFZ} & {\tt CPPU} \\
      \midrule
      4 &        1.023 &     1.012 &     807.79 &  1.19 \\
      6 &        1.052 &     1.018 &    1,052.39 & 1.29 \\
      8 &        1.029 &     1.028 &    4,625.46 & 1.12 \\
      \bottomrule
    \end{tabular}
    \caption{Approximation ratios and running times
      of our MR algorithm ({\tt CPPU}) and {\tt AFZ}.}
    \label{tab:comparison}
  \end{minipage}
  \hfill
  \begin{minipage}{.54\linewidth}
    \centering
    \includegraphics[width=\columnwidth]{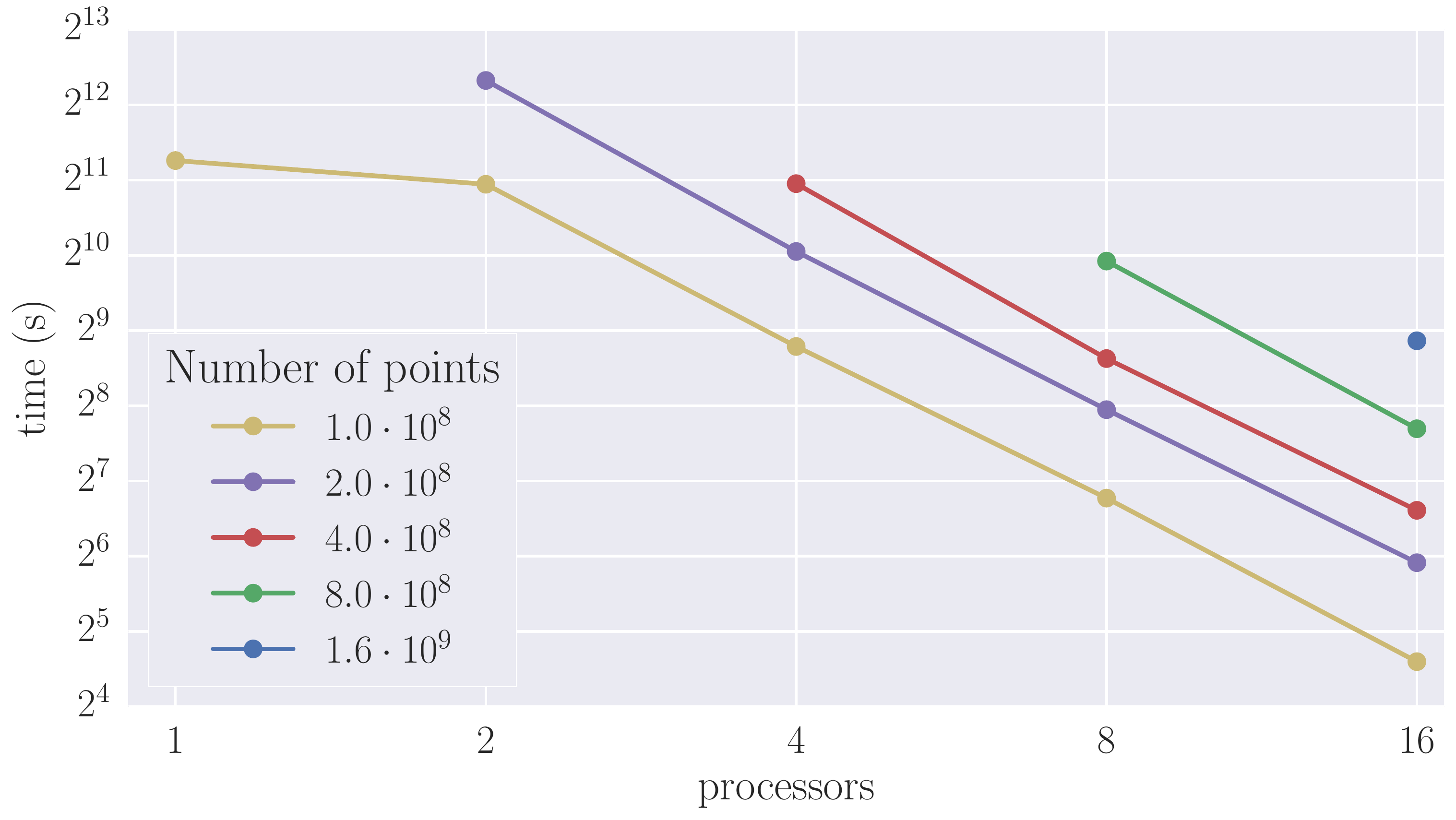}
    \captionof{figure}{Scalability of our algorithms for different number of points and processors. The running time for one processor is obtained with the streaming algorithm.}
    \label{fig:scalability}
  \end{minipage}
\end{table}

\subsection{Scalability}

We report on the scalability of our MR algorithm on datasets drawn
from $\mathbb{R}^3$, ranging from 100 million points (the same dataset
used in subsections~\ref{sec:streaming-experiments}
and~\ref{subsec:MRexp}) up to 1.6 billion points.
We fixed the size $s$ of the memory required by the final reducer and varied the number of processors used.
On a single machine, instead of running MapReduce, which makes little sense, we run the streaming algorithm with $k'=2048$, so to have a final coreset of the same size as the ones found in MapReduce runs.
For a given number of processors $p$ and number of points $n$, we run the corresponding experiment only if $n/p$ points fit into the main memory of a single processor.
As shown in Figure~\ref{fig:scalability}, for a fixed dataset size,
our MapReduce algorithm exhibits super-linear scalability: doubling the number of processors results in a 4-fold gain in running time (at the expense of a mild worsening of the approximation ratio, as pointed out in Subsection~\ref{subsec:MRexp}).
The reason is that each reducer performs $\BO{ns/(kp^2)}$ work to
build its core-set, where $p$ is the number of reducers, since the
core-set construction involves $s/(kp)$ iterations, with each
iteration requiring the scan of $n/p$ points.

For the dataset with 100 million points, the MR algorithm outperforms
the streaming algorithm in every processor configuration. It must be
remarked that the running time reported in
Figure~\ref{fig:scalability} for the streaming algorithm takes into
account also the time needed to stream data from main memory (unlike
the throughput reported in Figure~\ref{fig:throughput}). This is to
ensure a fair comparison with MapReduce, where we also take into
account the time needed to shuffle data between the first and the
second round, and the setup time of the rounds.  Also, we
note that the streaming algorithm appears to be faster than what the
MR algorithm would be if executed on a single processor, and this is
probably due to the fact that the former is more cache
friendly.

If we fix the number of processors, we observe that our algorithm
exhibits linear scalability in the number of points.
Finally, in a set of experiments, omitted for brevity, we
verified that for a fixed number of processors the
time increases linearly with $k'$. Both these behaviors
are in accordance with the theory.

\section{Acknowledgments}

Part of this work was done while the authors were
visiting the Departiment of Computer Science
at Brown University.
This work was supported, in part, by MIUR of Italy under
project AMANDA, and by the University of Padova under project
CPDA152255/15: "Resource-Tradeoffs Based Design
of Hardware and Software for Emerging Computing Platforms".
The work of Eli Upfal was supported in part by NSF grant IIS-1247581
and NIH grant R01-CA180776.

\bibliographystyle{abbrv}
\bibliography{references}
\end{document}